\documentclass[10pt,journal,epsfig]{IEEEtran}

\usepackage[dvips]{graphicx}
\usepackage{graphicx}
\usepackage{amssymb}
\usepackage{cite}
\usepackage{subfigure}
\usepackage{amsmath}
\usepackage{algorithm}
\usepackage{algorithmic}
\usepackage{multirow}
\usepackage{xcolor}
\usepackage{upgreek}
\usepackage{amsthm}
\allowdisplaybreaks[4]

\allowdisplaybreaks[4]

\begin{document}

\title{Max-Min Beamforming for Large-Scale Cell-Free Massive MIMO: A Randomized ADMM Algorithm}
\author{Bin Wang,
Jun Fang,~\IEEEmembership{Senior Member,~IEEE},
Yue Xiao,~\IEEEmembership{Member,~IEEE},
Martin Haardt,~\IEEEmembership{Fellow,~IEEE}
\thanks{Bin Wang, Jun Fang and Yue Xiao are with the National Key Laboratory
of Wireless Communications, University of
Electronic Science and Technology of China, Chengdu 611731, China,
Email: JunFang@uestc.edu.cn; Bin Wang is also with the Guangdong Dejiu New Energy Co., Ltd, Foshan 528200, China.}
\thanks{Martin Haardt is with the Communications Research
Laboratory, Ilmenau University of Technology, 98693,
Ilmenau, German, Email: martin.haardt@tu-ilmenau.de}
}

\maketitle

\begin{abstract}
We consider the problem of max-min beamforming (MMB) for cell-free
massive multi-input multi-output (MIMO) systems, where the objective
is to maximize the minimum achievable rate among all users. Existing
MMB methods are mainly based on deterministic optimization methods,
which are computationally inefficient when the problem size grows
large. To address this issue, we, in this paper, propose a
randomized alternating direction method of multiplier (ADMM)
algorithm for large-scale MMB problems. We first propose a novel
formulation that transforms the highly challenging
feasibility-checking problem into a linearly constrained
optimization problem. An efficient randomized ADMM is then
developed for solving the linearly constrained problem. Unlike
standard ADMM, randomized ADMM only needs to solve a small number
of subproblems at each iteration to ensure convergence, thus
achieving a substantial complexity reduction. Our theoretical
analysis reveals that the proposed algorithm exhibits an
$O(1/\bar{t})$ convergence rate ($\bar{t}$ represents the number
of iterations), which is on the same order as its deterministic
counterpart. Numerical results show that the proposed algorithm
offers a significant complexity advantage over existing methods in
solving the MMB problem.
\end{abstract}


\section{Introduction}
The last decade has witnessed a pivotal advance in mobile devices,
such as smartphones and notebooks. The proliferation of mobile
devices has triggered an unprecedented amount of mobile data
traffic. To satisfy the stringent need for high-speed
wireless data transmission, cell-free (CF) massive MIMO (multi-input
multi-output) is envisioned to play a central role in future
wireless systems \cite{ZhangBjornson20,YouWang19}. A CF
massive MIMO network consists of a number of Access Points (APs)
that are densely distributed over a wide area. All APs in the
network share the same time-frequency resources to provide
services for multiple users. A CPU (Central Processing Unit)
is connected to all APs to coordinate data transmission.

In a CF massive MIMO system, APs are usually densely located
in order to provide a high data transmission rate for users.
This, however, causes detrimental interference and presents
a major obstacle to maximizing the network capacity. An
essential way to suppress interference is to design highly
directional beams for different users. As such, developing
efficient beamforming techniques is vital for improving the
system performance. A popular method for beamforming is to
maximize the sum of the achievable rates of the users, which
leads to the sum-rate maximization problem. However, sum-rate
maximization usually ignores user fairness. Indeed, it is
often the case that users with high-quality channel conditions
receive a better quality-of-service (QoS) compared to those
with low-quality channel conditions. This is unacceptable in
many scenarios. An example is multicasting, where the common
information rate is determined by the worst-case data transmission
rate among users. The importance of fairness is also evident
in industrial wireless networks. In the Yangshan port of China,
for example, more than $100$ AGVs (automatic guided vehicles)
are controlled through a $5.8$GHz LTE (long term evolution)
network. In order to ensure the safe remote control of AGVs,
it is required that the transmission delay of the entire
network should not exceed $50$ms. In such a network, every
AGV should be provided with a sufficiently high QoS to
maintain the normal operation of the port. A common way
to enforce fairness is to employ either max-min
beamforming (MMB) or QoS-aware beamforming. The max-min
beamforming aims to maximize the minimum user achievable
rate while the QoS-aware beamforming sets a QoS threshold
that must be met for each user. Both of them aim to provide
uniformly good QoS for all users. These two beamforming
strategies have been studied in a variety of works, e.g.
\cite{XieXu20,MengHu23,NgoAshikhmin17,SadeghiBjornson17,
WuZhang19a,WuZhang19b,ChenTao17}. In this paper, we primarily
focus on the MMB problem. Nevertheless, we will show that
the proposed algorithm can also be extended to the QoS-aware
beamforming problem.

Due to the max-min form and the log-fractional structure
of the objective function, solving the MMB problem is
challenging. Existing methods for the MMB
problem can be roughly divided into two categories: dual
methods \cite{SchubertBoche04,HuangTan13,GongJordan09,YuLan07,
BjornsonJorswieck13,MirettiCavalcante24,MirettiCavalcante23,Nuzman07}
and primal methods \cite{XieXu20,ChakrabortyDemir20,ShenYu18,
NgoAshikhmin17,ShiZhang15,WangWang21,CaiWang23}. For a
centralized network with a total power constraint, dual
methods can employ the uplink-downlink duality (UDD) principle
to transform the highly challenging downlink beamforming
problem into an equivalent uplink combining problem, which
is much easier to solve since it has a favorable separable
structure that allows optimization variables decoupled
and updated in an alternating fashion. Based on the UDD
principle, the work \cite{YuLan07,BjornsonJorswieck13} has
derived the uplink dual problem for general MIMO networks
with different kinds of power constraints. The more recent
work \cite{MirettiCavalcante24}, which shares a similar
UDD principle with that of \cite{YuLan07,BjornsonJorswieck13},
extended the UDD-based method to more involved setups such
as constrained backhaul capacities as well as precoding
with only local CSI (channel state information). Nevertheless,
for the cell-free massive MIMO systems where each AP has
its own transmit power constraint, the uplink dual problem
becomes a highly coupled semi-definite program (SDP)
\cite{YuLan07}, which also calls for advanced optimization
methods. In addition to dual methods, primal methods
target directly on the original downlink MMB problem by
employing optimization tools such as semidefinite relaxation
(SDR) \cite{XieXu20}, fractional transformation
\cite{ChakrabortyDemir20,ShenYu18}, and bisection search
\cite{NgoAshikhmin17,XieXu20,ShiZhang15}. Among them,
SDR and fractional transformation are used to simplify
the log-fractional achievable rate function such that
the resulting problem can be easily solved using classical
tools such as interior-point methods
\cite{NesterovNemirovskii94} or ADMM (alternating direction
method of multipliers). However, these methods either
lift the problem to a much higher dimensional space or
necessitate alternating optimization between variables,
both of which are computationally expensive. Moreover,
these methods usually lack a global optimality guarantee.
The bisection search methods (BSM) can alleviate the above
drawbacks to some extent. The basic idea of BSM is to
transform the MMB problem into a sequence of convex
conic feasibility-checking subproblems (CCFCPs). In
this regard, efficiently solving the CCFCP is the key
to the success of BSM. In recent years, many
efforts have been devoted to investigating the CCFCP
(or the more general feasibility-checking problems)
\cite{BauschkeCombettes04,BauschkeMoursi17,
CensorElfving05,OdonoghueChu16,YeTodd94,XuHuang96}. In
particular, the Douglas-Rachford splitting (DRS) method
\cite{BauschkeCombettes04,BauschkeMoursi17} was employed
in \cite{WangWang21,CaiWang23} to address the MMB problem
in cell-free or collocated MIMO systems. Nevertheless,
these first-order methods still face challenges for
large-scale massive MIMO with an ever-increasing problem size.

In this paper, we propose a randomized ADMM algorithm to
efficiently solve the MMB problem in the cell-free massive
MIMO network. Our contribution mainly lies in two aspects.
Firstly, we propose a novel formulation that transforms the
highly challenging large-scale feasibility-checking problem
into a linearly constrained optimization problem. The objective
function of the linearly constrained problem is convex and
smooth. The second contribution is that we propose a highly
efficient randomized ADMM to solve the linearly constrained
optimization problem. By utilizing the orthogonal structure
inherent in the linear constraint, the most computationally
expensive subproblem is divided into a number of independent
subtasks. In each iteration, only a small number of subtasks
needs to be solved, thus leading to a much lower per-iteration
computational complexity compared to existing methods. In
the theoretical aspect, we prove that the proposed algorithm
exhibits an $O(1/\bar{t})$ convergence rate ($\bar{t}$
represents the number of iterations), which is in the same
order as other state-of-the-art deterministic methods.
Numerical results show that the proposed algorithm presents
a significant complexity advantage over existing methods in
solving the MMB problem.

The rest of this paper is organized as follows. In Section
\ref{sec-problem}, we introduce the communication model and
formulate the max-min beamforming (MMB) problem. In Section
\ref{sec-linear}, the feasibility-checking problem is formulated
into a linearly constrained problem. A standard ADMM and a randomized ADMM are developed
in Section \ref{sec-standard-ADMM} and Section \ref{sec-proposed}, respectively to solve
the linearly constrained problem.
The convergence result of the proposed algorithm is provided in
Section \ref{sec-convergence}. Section \ref{sec-extension-QoS} discusses
how to extend our method to QoS-aware beamforing problems. Section \ref{sec-relation}
discusses the relation between the proposed method and
existing methods. Simulations results are provided in
Section \ref{sec-simulation}, followed by concluding remarks
in Section \ref{sec-conclusion}.

\section{Problem formulation}
\label{sec-problem}

\subsection{Communication Model}
\label{subsec-comm-net} Consider a CF massive MIMO system with $M$
APs (access points) and $K$ single-antenna users. We assume that
there exists a CPU (central processing unit) which is responsible
for coordinating the data transmission for all users. The CPU is
connected to each AP through a backhaul link. For simplicity, we
assume that each AP is equipped with $N$ antennas. In this paper,
we consider the downlink data transmission. We assume that perfect
channel state information can be obtained at the CPU. The channel
between the $m$th AP and user $k$ is denoted as $\boldsymbol{h}_k[m]
\in\mathbb{C}^{N}$. The signal intended for user $k$ can be written
as
\begin{align}
y_k=&\textstyle\sum_{m=1}^M (\boldsymbol{h}_k[m])^H\boldsymbol{v}_k[m]
 u_k+
\nonumber\\
&\textstyle\sum_{k'\neq k}\sum_{m=1}^{M}(\boldsymbol{h}_{k}[m])^H
\boldsymbol{v}_{k'}[m]u_{k'}+n_k, \forall k,
\label{received-signal}
\end{align}
where $u_k$ is the symbol intended for user $k$,
$\boldsymbol{v}_k[m]\in\mathbb{C}^{N}$ is the beamforming vector
designed for the $k$th user at the $m$th AP, and $n_k\sim
\mathcal{CN}(0,\sigma_k^2)$ is the additive Gaussian noise. To
simplify notations, denote $\boldsymbol{h}_{k}\triangleq
[(\boldsymbol{h}_k[1])^T\cdots(\boldsymbol{h}_k[M])^T]^T \in
\mathbb{C}^{MN}$ as the joint channel vector between user $k$ and
all APs. Also denote $\boldsymbol{v}_{k}\triangleq
[(\boldsymbol{v}_k[1])^T\cdots(\boldsymbol{v}_k[M])^T]^T\in
\mathbb{C}^{MN}$ as the joint beamforming vector for user $k$. The
achievable rate of user $k$ can therefore be given as
\begin{align}
& R_k(\{\boldsymbol{v}_{k'}\}_{k'=1}^K)=
\text{log}\Big(1+\frac{|\boldsymbol{h}_k^H\boldsymbol{v}_k|^2}
{\sum_{k'\neq k} |\boldsymbol{h}_k^H\boldsymbol{v}_{k'}|^2+
\sigma_k^2}\Big), \ \forall k.
\label{def-rate-cran}
\end{align}

\subsection{Max-Min Beamforming}
We consider the following max-min beamforming problem:
\begin{align}
\textbf{\text{MMB}}: \
\mathop {\max }_{\{\boldsymbol{v}_{k}\}_{1\leq k\leq K}}
 &\mathop {\min }_{1\leq k\leq K} \
R_k(\{\boldsymbol{v}_{k'}\}_{1\leq k'\leq K})
\nonumber\\
\text{s.t.} & \textstyle \
\sum_{k=1}^K ||\boldsymbol{v}_k[m]||_2^2\leq p_m,
1\leq m\leq M,
\label{maxmin-beam}
\end{align}
where $p_m$ is the maximum transmit power of the $m$th AP.
The above formulation aims to maximize the minimum rate
among all users.

%Note that our proposed algorithm can also
%be extended to other types of power constraints. The will
%be elaborated in Section \ref{sec-linear-d}.

\begin{algorithm}
\caption{The bisection method} \label{alg:bsm}
\begin{algorithmic}
\STATE{\textbf{Inputs}: $s_{\text{min}}$, $s_{\text{max}}$ and
$s_{\text{ter}}$. All initial vectors are set to $\boldsymbol{0}$}.
\STATE{\textbf{While} $s_{\text{max}}-s_{\text{min}}
\leq s_{\text{ter}}$
\textbf{do}}
\STATE{\ \ \ \ \textcircled{1} \ Set $s=s_c\triangleq\frac{s_{\text{min}}+
s_{\text{max}}}{2}$ and then employing an algorithm to solve problem
(\ref{feasibility-problem}).}
\STATE{\ \ \ \ \textcircled{2} \ If (\ref{feasibility-problem})
is feasibile (resp. infeasibile), update $s_{\text{min}}=s_c$
(resp. $s_{\text{max}}=s_c$) ; }
\par{\textbf{End While and Output}
$\{\boldsymbol{v}_{k'}\}_{1\leq k'\leq K}$;}
%\par{\textbf{Outputs:} the beamforming vector $\{\boldsymbol{v}_{k'}\}_{1\leq k'\leq K}$.}
\end{algorithmic}
\end{algorithm}

\subsection{The Bisection Framework}
Due to the complex structure of
$R_k(\{\boldsymbol{v}_{k'}\}_{1\leq k'\leq K})$, directly solving
problem (\ref{maxmin-beam}) is difficult. A commonly used method
is the bisection method (BSM), which is summarized in Algorithm
\ref{alg:bsm}. To understand BSM, consider the following problem:
\begin{align}
\textbf{\text{BSM}}: \
\mathop {\max }\limits_{\{\boldsymbol{v}_{k}\}_{1\leq k\leq K},\ s}
& \ s
\nonumber\\
\text{s.t.} & \textstyle \
\sum_{k=1}^K ||\boldsymbol{v}_k[m]||_2^2\leq p_m, \ 1\leq m\leq M,
\nonumber\\
& \ R_k(\{\boldsymbol{v}_{k'}\}_{1\leq k'\leq K})\geq s, \forall k.
\label{bsm-problem}
\end{align}
where $s$ is an auxiliary variable. Clearly, problem
(\ref{bsm-problem}) is equivalent to (\ref{maxmin-beam}).
The idea of BSM is to use bisection search to find the optimal
$s$. In the beginning of Algorithm \ref{alg:bsm}, a sufficiently
large range for $s$ is set, i.e., $[s_{\text{min}},s_{\text{max}}]$.
In each iteration, let $s$ be fixed as $s_c=(s_{\text{min}}
+s_{\text{max}})/2$, (\ref{bsm-problem}) becomes
\begin{align}
\textbf{\text{FC}}: \
\text{Find} & \ \{\boldsymbol{v}_{k}\}_{k=1}^k
\nonumber\\
\text{s.t.} & \textstyle \
\sum_{k=1}^K ||\boldsymbol{v}_k[m]||_2^2\leq p_m, \ 1\leq m\leq M,
\nonumber\\
& \ R_k(\{\boldsymbol{v}_{k'}\}_{k=1}^k)\geq s_{c},
\forall k.
\label{feasibility-problem}
\end{align}
If (\ref{feasibility-problem}) is feasible (resp. infeasible),
$s_{\text{min}}$ (resp. $s_{\text{max}}$) is updated as $s_c$.
This process is repeated until some stopping criterion is met,
say $s_{\text{max}}-s_{\text{min}}<s_{\text{ter}}$, where
$s_{\text{ter}}$ is a prescribed value. Clearly, the key to
solving (\ref{maxmin-beam}) is to efficiently solve each
feasibility-checking (FC) problem (\ref{feasibility-problem}).

\section{A Linearly Constrained Formulation for the Feasibility-Checking Problem}
\label{sec-linear} In this section, we convert the highly
challenging feasibility-checking problem (\ref{feasibility-problem})
into a linearly constrained optimization problem which is much more amiable to solve.

%We will show that the objective function is continuously
%differentiable. Moreover, we will also show that the data
%matrices in the linear constraint enjoys favorable mutual
%orthogonality, which plays an important role in developing
%the randomized ADMM algorithm.

\subsection{Preliminaries}
A set $\mathcal{C}$ is said to be a cone if $\boldsymbol{x}\in
\mathcal{C}\Rightarrow t\boldsymbol{x}\in\mathcal{C}$,
$\forall t\geq0$. Examples of cones include $\mathbb{R}^n$,
$\mathbb{R}^n_{+}$, and $\mathbb{S}_{+}^{n\times n}$, where
$\mathbb{S}_{+}^{n\times n}$ is the set of all $n\times n$
positive semidefinite matrices. Throughout this paper, we
will frequently use the so-called second-order cone (SOC)
defined as
\begin{align}
\mathcal{C}(\tau)\triangleq\{[\boldsymbol{x}^T \
y]^T\in\mathbb{R}^n\times \mathbb{R}_{+} \ | \
\|\boldsymbol{x}\|_2\leq \tau\cdot y\}, \label{def-soc}
\end{align}
where $\tau$ is a given nonnegative scalar. Notably, the
SOC is convex and closed.

\subsection{Problem Formulation}
\label{sec-linear-b}
First notice that $R_k(\{\boldsymbol{v}_{k'}\}_{1\leq k'\leq K})
\geq s_c$ in (\ref{feasibility-problem}) can be equivalently
written as
\begin{align}
&\textstyle |\boldsymbol{h}_k^H\boldsymbol{v}_k|^2 \geq
(2^{s_c}-1)\big(\sum_{k'\neq k}
|\boldsymbol{h}_k^H\boldsymbol{v}_{k'}|^2+\sigma_k^2\big)
\nonumber\\
&\Rightarrow\textstyle
2^{s_c}|\boldsymbol{h}_k^H\boldsymbol{v}_k|^2 \geq (2^{s_c}-1)
\big(\sum_{k'=1}^K
|\boldsymbol{h}_k^H\boldsymbol{v}_{k'}|^2 +\sigma_k^2\big).
\label{problem-form-2}
\end{align}
Substituting (\ref{problem-form-2}) into (\ref{feasibility-problem})
yields
\begin{align}
\textbf{\text{FC}}: \ &\text{Find}  \ \{\boldsymbol{v}_{k}\}_{1\leq k\leq K}
\nonumber\\
\text{s.t.} & \textstyle \
\sum_{k=1}^K ||\boldsymbol{v}_k[m]||_2^2\leq p_m, \ 1\leq m\leq M,
\nonumber\\
&\textstyle \ e(s_c)\cdot|\boldsymbol{h}_k^H\boldsymbol{v}_k|^2
\geq \big(\sum_{k'=1}^K |\boldsymbol{h}_k^H\boldsymbol{v}_{k'}|^2
+\sigma_k^2\big), \ \forall k,
\label{problem-form-3}
\end{align}
where $e(s_c)\triangleq \frac{2^{s_c}}{2^{s_c}-1}$ is a function
of $s_c$. Problem (\ref{problem-form-3}) enjoys the so-called
\emph{\textbf{phase ambiguity property}}. This means that, if
there exists a set of $\{\boldsymbol{v}_{k}\}_{1\leq k\leq K}$
such that (\ref{problem-form-3}) is feasible, then
$\{\hat{\boldsymbol{v}}_{k}\triangleq\exp\{-j\theta_k\}
\cdot\boldsymbol{v}_k, 0\leq\theta_k\leq 2\pi\}_{1\leq k\leq K}$
is also a feasible solution for (\ref{problem-form-3}). Based on
this equivalence, we have the following proposition.

\newtheorem{proposition}{Proposition}
\begin{proposition}
\label{theorem-content-1}
Solving problem (\ref{problem-form-3}) is equivalent to solving
the following problem:
\begin{align}
\textbf{\emph{\text{FC}}}: \ & \emph{\text{Find}} \
\{\boldsymbol{v}_k\}_{1\leq k\leq K}
\nonumber\\
\text{s.t.} & \textstyle \
\sum_{k=1}^K ||\boldsymbol{v}_k[m]||_2^2\leq p_m, \ 1\leq m\leq M,
\nonumber\\
&\textstyle \ e(s_c)^{\frac{1}{2}}
\cdot\emph{\text{Re}}\{\boldsymbol{h}_k^H\boldsymbol{v}_k\}
\geq \big(\sum_{k'=1}^K |\boldsymbol{h}_k^H\boldsymbol{v}_{k'}|^2
+\sigma_k^2\big)^{\frac{1}{2}}, \ \forall k.
\label{problem-form-4}
\end{align}
\end{proposition}

Since problem (\ref{problem-form-3}) and (\ref{problem-form-4})
are equivalent, we now focus on problem (\ref{problem-form-4}).
Before proceeding, we introduce several notations to
facilitate our subsequent derivations. Define
\begin{align}
&\boldsymbol{H}_k\triangleq
\begin{bmatrix}
\tilde{\boldsymbol{H}}_k &\cdots  &\boldsymbol{0}_{2\times 2MN} & \cdots &\boldsymbol{0}_{2\times 2MN}  \\
\vdots  & \ddots  & \vdots  &\vdots &\vdots  \\
\boldsymbol{0}_{2\times 2MN}  & \cdots  & \tilde{\boldsymbol{H}}_k  &\cdots &\boldsymbol{0}_{2\times 2MN}  \\
\vdots  & \cdots & \vdots  &\ddots & \vdots\\
\boldsymbol{0}_{2\times 2MN}  & \cdots &\boldsymbol{0}_{2\times 2MN} &\cdots & \tilde{\boldsymbol{H}}_k\\
\boldsymbol{0}_{1\times 2MN}  & \cdots &\boldsymbol{0}_{1\times 2MN} & \cdots & \boldsymbol{0}_{1\times 2MN} \\
\boldsymbol{0}_{1\times 2MN}  & \cdots  &e(s_c)^{1/2}\cdot\tilde{\boldsymbol{h}}_{k}^T  &\cdots & \boldsymbol{0}_{1\times 2MN}
\end{bmatrix},
\nonumber\\
&\tilde{\boldsymbol{v}}\triangleq[
\tilde{\boldsymbol{v}}_1^T \ \cdots \
\tilde{\boldsymbol{v}}_K^T]^T, \
\boldsymbol{b}_k\triangleq
[\boldsymbol{0}_{1\times 2}  \ \cdots  \ \boldsymbol{0}_{1\times 2}  \ \sigma_k  \ 0]^T, \ k\leq K,
\label{problem-form-4-1}
\end{align}
where $\tilde{\boldsymbol{h}}_{k}^T\triangleq
[\text{Re}\{\boldsymbol{h}_k\}^T \ -\text{Im}\{\boldsymbol{h}_k\}^T]$,
\begin{align}
&\tilde{\boldsymbol{H}}_k\triangleq
\begin{bmatrix}
\text{Re}\{\boldsymbol{h}_k\}^T & -\text{Im}\{\boldsymbol{h}_k\}^T\\
\text{Im}\{\boldsymbol{h}_k\}^T & \text{Re}\{\boldsymbol{h}_k\}^T
\end{bmatrix},
\tilde{\boldsymbol{v}}_k\triangleq
\begin{bmatrix}
\text{Re}\{\boldsymbol{v}_k\} \\
\text{Im}\{\boldsymbol{v}_k\}
\end{bmatrix}.
\label{problem-form-4-2}
\end{align}
Notice that there are $K$ block columns in $\boldsymbol{H}_k$
and $e(s_c)^{1/2}\cdot\tilde{\boldsymbol{h}}_{k}^T$ lies
in the $k$th block column. Recall that $\tilde{\boldsymbol{H}}_k$
is a $2\times 2MN$ matrix, thus $\boldsymbol{H}_k$ is
a $(2K+2)\times 2MNK$ matrix. Meanwhile, $\tilde{\boldsymbol{v}}_k$
is a length-$2MN$ vector and $\boldsymbol{b}_k$ is a
length-$(2K+2)$ vector. Based on the above notations, we
have
\begin{align}
&\tilde{\boldsymbol{H}}_k\tilde{\boldsymbol{v}}_{k'}=
\begin{bmatrix}
\text{Re}\{\boldsymbol{h}_k^H\boldsymbol{v}_{k'}\} &
\text{Im}\{\boldsymbol{h}_k^H\boldsymbol{v}_{k'}\}
\end{bmatrix}^T,
\nonumber\\
&\tilde{\boldsymbol{h}}_{k}^T\tilde{\boldsymbol{v}}_{k'}
=\text{Re}\{\boldsymbol{h}_k^H\boldsymbol{v}_{k'}\}, \
1\leq k'\leq K.
\label{problem-form-4-2-1}
\end{align}
With these notations, the second line in (\ref{problem-form-4})
can be equivalently written as
$\boldsymbol{H}_k\tilde{\boldsymbol{v}} +\boldsymbol{b}_k
\in\mathcal{C}(e(s_c)^{\frac{1}{2}})$, where $\mathcal{C}(\cdot)$
is an SOC defined in (\ref{def-soc}). To see the equivalence,
note that
\begin{align}
&\boldsymbol{H}_k\tilde{\boldsymbol{v}}+\boldsymbol{b}_k=
[(\tilde{\boldsymbol{H}}_k\tilde{\boldsymbol{v}}_1)^T \ \cdots \
(\tilde{\boldsymbol{H}}_k\tilde{\boldsymbol{v}}_K)^T  \ \ \sigma_k \ \
e(s_c)^{\frac{1}{2}}\tilde{\boldsymbol{h}}_{k}^T\tilde{\boldsymbol{v}}_k]^T
\nonumber\\
&\overset{(\ref{problem-form-4-2-1})}{=}
[\text{Re}\{\boldsymbol{h}_k^H\boldsymbol{v}_1\} \
\text{Im}\{\boldsymbol{h}_k^H\boldsymbol{v}_1\} \ \cdots \
\text{Re}\{\boldsymbol{h}_k^H\boldsymbol{v}_K\} \
\text{Im}\{\boldsymbol{h}_k^H\boldsymbol{v}_K\}
\nonumber\\
&\qquad\quad\sigma_k \ \
e(s_c)^{\frac{1}{2}}\text{Re}\{\boldsymbol{h}_k^H\boldsymbol{v}_k\}]^T.
\label{problem-form-4-2-2}
\end{align}
Based on (\ref{problem-form-4-2-2}), if $\tilde{\boldsymbol{v}}$
satisfies $\boldsymbol{H}_k\tilde{\boldsymbol{v}}+\boldsymbol{b}_k
\in\mathcal{C}(e(s_c)^{\frac{1}{2}})$, then we have
\begin{align}
&\textstyle\big((\sum_{k'=1}^K\text{Re}\{\boldsymbol{h}_k^H
\boldsymbol{v}_{k'}\}^2+\text{Im}\{\boldsymbol{h}_k^H
\boldsymbol{v}_{k'}\}^2\})+\sigma_k^2\big)^{\frac{1}{2}}
\nonumber\\
\leq&
e(s_c)^{\frac{1}{2}}\cdot\text{Re}\{\boldsymbol{h}_k^H\boldsymbol{v}_k\},
\end{align}
which is exactly the QoS constraint in (\ref{problem-form-4}).
For simplicity, hereafter we use $\mathcal{C}$ to refer to
$\mathcal{C}(e(s_c)^{\frac{1}{2}})$. Using the above equivalence,
problem (\ref{problem-form-4}) can be written as
\begin{align}
\textbf{\text{FC}}: \ \text{Find} & \quad \tilde{\boldsymbol{v}}
\nonumber\\
\text{s.t.} & \quad \boldsymbol{H}_k\tilde{\boldsymbol{v}}+
\boldsymbol{b}_k\in\mathcal{C}, \forall k, \
\nonumber\\
&\quad\textstyle
\sum_{k=1}^K ||\boldsymbol{v}_k[m]||_2^2\leq p_m, \ 1\leq m\leq M,
%\nonumber\\
%&\quad \|\tilde{\boldsymbol{v}}\|_2^2\leq p.
\label{feasibility-compact}
\end{align}
To rewrite the constraints in a more compact form, define a convex
set $\mathcal{D}\triangleq\mathcal{C}\times\cdots\times \mathcal{C}
\times \mathcal{P}_1 \times \cdots \times\mathcal{P}_M$, where
$\times$ denotes the Cartesian product, and also define
\begin{align}
&\textstyle\mathcal{P}_m\triangleq\Big\{\{\boldsymbol{v}_k[m]\}_{k=1}^K \ \Big| \
\sum_{k=1}^K ||\boldsymbol{v}_k[m]||_2^2\leq p_m\Big\},
\label{feasibility-constraints}
\end{align}
Then the constraints in (\ref{feasibility-compact}) can be equivalently
written as
\begin{align}
\boldsymbol{H}\tilde{\boldsymbol{v}}+
\boldsymbol{b}\in\mathcal{D}
\label{feasibility-constraints-1}
\end{align}
where
\begin{align}
&\boldsymbol{H}\triangleq
[\boldsymbol{H}_1^T \ \cdots \ \boldsymbol{H}_K^T \
\boldsymbol{P}_{ap}
]^T\in\mathbb{R}^{(K(2K+2)+2MNK)\times 2MNK},
\nonumber\\
&\boldsymbol{b}\triangleq
[\boldsymbol{b}_1^T \ \cdots \ \boldsymbol{b}_K^T \
\boldsymbol{0}_{2MKN\times 1}^T
]^T\in\mathbb{R}^{K(2K+2)+2MNK},
\label{feasibility-compact-1-1}
\end{align}
In the above, $\boldsymbol{P}_{ap}$ is a permutation matrix
($\boldsymbol{P}_{ap}^T\boldsymbol{P}_{ap}=\boldsymbol{I}$ and
$\boldsymbol{P}_{ap}\boldsymbol{P}_{ap}^T=\boldsymbol{I}$)
that takes $\tilde{\boldsymbol{v}}$ as the input and then outputs
$\breve{\boldsymbol{v}}$, namely, $\boldsymbol{P}_{ap}\tilde{\boldsymbol{v}}
=\breve{\boldsymbol{v}}$. Note that $\breve{\boldsymbol{v}}$
is a permuted vector of $\tilde{\boldsymbol{v}}$ defined as
\begin{align}
\breve{\boldsymbol{v}}=
\begin{bmatrix}
\breve{\boldsymbol{v}}_1 \\
\vdots \\
\breve{\boldsymbol{v}}_M
\end{bmatrix}, \
\breve{\boldsymbol{v}}_m=
\begin{bmatrix}
\text{Re}\{\boldsymbol{v}_1[m]\} \\
\text{Im}\{\boldsymbol{v}_1[m]\} \\
\vdots \\
\text{Re}\{\boldsymbol{v}_K[m]\} \\
\text{Im}\{\boldsymbol{v}_K[m]\} \\
\end{bmatrix}, \ 1\leq m\leq M.
\end{align}
Clearly, $\breve{\boldsymbol{v}}_m\in\mathcal{P}_m$ corresponds
to the $m$th power constraint in (\ref{feasibility-compact}).

Replacing the constraints in (\ref{feasibility-compact})
with (\ref{feasibility-constraints-1}), we obtain
\begin{align}
\textbf{\text{FC}}: \ \text{Find} &\quad \tilde{\boldsymbol{v}}
\nonumber\\
\text{s.t.} &\quad  \boldsymbol{H}\tilde{\boldsymbol{v}}+
\boldsymbol{b}\in\mathcal{D}.
\label{feasibility-compact2}
\end{align}

\subsection{Linearly Constrained Formulation}
Solving problem (\ref{feasibility-compact2}) is equivalent to
solving the following unconstrained problem:
\begin{align}
\textstyle\mathop {\min}\limits_{\tilde{\boldsymbol{v}}}  \
f(\boldsymbol{H}\tilde{\boldsymbol{v}}+\boldsymbol{b})
\triangleq\frac{1}{2}\|\boldsymbol{H}\tilde{\boldsymbol{v}}
+\boldsymbol{b}-\text{Proj}_{\mathcal{D}}\{\boldsymbol{H}
\tilde{\boldsymbol{v}}+\boldsymbol{b}\}\|_2^2,
\label{problem-linear}
\end{align}
where $\text{Proj}_{\mathcal{D}}\{\boldsymbol{x}\}$ is the
projection operator onto $\mathcal{D}$. In particular, we have
$\text{Proj}_{\mathcal{D}}\{\boldsymbol{x}\}=\boldsymbol{x}$
if $\boldsymbol{x}\in\mathcal{D}$, and
\begin{align}
\text{Proj}_{\mathcal{D}}\{\boldsymbol{x}\}=
\arg\min\limits_{\bar{\boldsymbol{x}}\in\mathcal{D}}
\|\bar{\boldsymbol{x}}-\boldsymbol{x}\|_2
\end{align}
if otherwise. Intuitively, if there exists a $\tilde{\boldsymbol{v}}$
such that $\boldsymbol{H}\tilde{\boldsymbol{v}}+\boldsymbol{b}$
lies in the set $\mathcal{D}$, then the minimum value of
$f(\boldsymbol{H}\tilde{\boldsymbol{v}}+\boldsymbol{b})$ is
$0$. On the contrary, if there does not exist a
$\tilde{\boldsymbol{v}}$ such that
$\boldsymbol{H}\tilde{\boldsymbol{v}} +\boldsymbol{b}$ lies
in $\mathcal{D}$, then the minimum value of
$f(\boldsymbol{H}\tilde{\boldsymbol{v}}+\boldsymbol{b})$
must be larger than $0$. As such, the minimum of
$f(\boldsymbol{H}\tilde{\boldsymbol{v}}+\boldsymbol{b})$
provides a certificate of whether the problem
(\ref{feasibility-compact2}) is feasible. Nevertheless,
since $f(\boldsymbol{H}\tilde{\boldsymbol{v}}+\boldsymbol{b})$
involves a nonlinear projection operator, it remains to
show that the minimum of $f(\boldsymbol{H}
\tilde{\boldsymbol{v}}+\boldsymbol{b})$ is attainable.
Hopefully, the following proposition shows that
$f(\boldsymbol{H}\tilde{\boldsymbol{v}} +\boldsymbol{b})$
is convex and continuously differentiable, which means
that the minimum of $f(\boldsymbol{H}\tilde{\boldsymbol{v}}+
\boldsymbol{b})$ is indeed attainable.

\begin{proposition}
(\cite{AubinCellina12,CensorElfving06}) The objective function
$f(\boldsymbol{x})$ in (\ref{problem-linear}) is convex and
continuously differentiable. The gradient of $f(\boldsymbol{x})$
is given by
\begin{align}
\nabla f(\boldsymbol{x})=\boldsymbol{x}-
\emph{\text{Proj}}_{\mathcal{D}}\{\boldsymbol{x}\}.
\label{Proposition-proof1}
\end{align}
Moreover, $\nabla f(\boldsymbol{x})$ is $1$-Lipschitz continuous.
\end{proposition}

To solve problem (\ref{problem-linear}), we introduce an auxiliary
variable $\boldsymbol{w}\triangleq[\boldsymbol{w}_1^T \ \cdots \
\boldsymbol{w}_{K+1}^T]^T$ to convert it into the following form:
\begin{align}
\textbf{\text{LCP}}: \
\mathop {\min}\limits_{\tilde{\boldsymbol{v}},\boldsymbol{w}} & \
\textstyle f(\boldsymbol{w})=\frac{1}{2}\|\boldsymbol{w}-
\text{Proj}_{\mathcal{D}}\{\boldsymbol{w}\}\|_2^2
\nonumber\\
\text{s.t.} & \ \boldsymbol{A}\tilde{\boldsymbol{v}}+\boldsymbol{b}
=\boldsymbol{w},
\label{problem-linear-cons}
\end{align}
where \textbf{LCP} is short for linearly constrained problem. To avoid
notational confusion in subsequent derivations, we have used
$\boldsymbol{A}$ to represent $\boldsymbol{H}$. Notice that
$\boldsymbol{A}$ can be horizontally split into $K$ submatrices,
that is, $\boldsymbol{A}=[\boldsymbol{A}_1\cdots\boldsymbol{A}_K]$,
in which
\begin{align}
\boldsymbol{A}_j\triangleq&
\begin{bmatrix}
[\boldsymbol{H}_1]_j^T &  \cdots &  [\boldsymbol{H}_K]_j^T &
[\boldsymbol{P}_{ap}]_j
\end{bmatrix}^T
\nonumber\\
&\in\mathbb{R}^{(K(2K+2)+2MNK)\times 2MN}, \ 1\leq j \leq K,
\label{problem-linear-cons-1-1}
\end{align}
and $[\boldsymbol{H}_k]_j$ (resp. $[\boldsymbol{P}_{ap}]_j$)
is the $j$th block column of $\boldsymbol{H}_k$ (resp.
$\boldsymbol{P}_{ap}$), see (\ref{problem-form-4-1}). Clearly,
$[\boldsymbol{H}_k]_j$ (resp. $[\boldsymbol{P}_{ap}]_j$) is a
$(2K+2)\times 2MN$ (resp. $2MKN\times 2MN$) matrix. Recalling
that $\boldsymbol{P}_{ap}$ is a permutation matrix, we have
$[\boldsymbol{P}_{ap}]_j^T[\boldsymbol{P}_{ap}]_{j'}=
\boldsymbol{0}$, $j'\neq j$. Thus we know that $\boldsymbol{A}_j$
and $\boldsymbol{A}_{j'}$ are orthogonal to each other,
$\forall j\neq j'$, which means
$\boldsymbol{A}_j^T\boldsymbol{A}_{j'} =\boldsymbol{0}$. Such
an orthogonality property plays a key role in the design of our
algorithm. With $\boldsymbol{A}=
[\boldsymbol{A}_1\cdots\boldsymbol{A}_K]$,
(\ref{problem-linear-cons}) can be equivalently written as
\begin{align}
\textbf{\text{LCP}}: \
\mathop {\min}\limits_{\tilde{\boldsymbol{v}},\boldsymbol{w}} & \
\textstyle f(\boldsymbol{w})=\frac{1}{2}\|\boldsymbol{w}-
\text{Proj}_{\mathcal{D}}\{\boldsymbol{w}\}\|_2^2
\nonumber\\
\text{s.t.} & \ \textstyle\sum_{j=1}^K
\boldsymbol{A}_j\tilde{\boldsymbol{v}}_j+\boldsymbol{b}=
\boldsymbol{w}.
\label{problem-linear-cons2}
\end{align}

\section{Standard ADMM-Based Method}
\label{sec-standard-ADMM} In this section, we introduce a
standard ADMM algorithm to solve (\ref{problem-linear-cons2}).
By analyzing the computational complexity of this algorithm,
we will see that its per-iteration computational cost is high,
especially for large-scale systems. This motivates us to
develop a randomized ADMM method in the next section.

\subsection{Standard ADMM for Solving (\ref{problem-linear-cons2})}
\label{sec-standard-ADMM-sub1}
The augmented Lagrangian function of (\ref{problem-linear-cons2})
is given as
\begin{align}
 L(\boldsymbol{w},\tilde{\boldsymbol{v}},\boldsymbol{\lambda})
=&\textstyle f(\boldsymbol{w})+\langle\boldsymbol{\lambda},
\sum_{j=1}^K \boldsymbol{A}_j\tilde{\boldsymbol{v}}_j
+\boldsymbol{b}-\boldsymbol{w}\rangle+
\nonumber\\
&\textstyle\frac{\beta}{2}\|\sum_{j=1}^K
\boldsymbol{A}_j\tilde{\boldsymbol{v}}_j
+\boldsymbol{b}-\boldsymbol{w}\|_2^2,
\end{align}
where $\boldsymbol{\lambda}\in\mathbb{R}^{(2K+2)K+2MNK}$ is the
Lagrangian multiplier, and $\beta$ is a parameter that needs to
be tuned. Based on $L(\boldsymbol{w},\tilde{\boldsymbol{v}},
\boldsymbol{\lambda})$, we can easily deduce a standard ADMM
for (\ref{problem-linear-cons2}):
\begin{align}
&\textbf{\text{Standard ADMM}}:
\nonumber\\
&\left \{
\begin{array}{ll}
\textstyle\tilde{\boldsymbol{v}}^{t}
=\arg\min\limits_{\tilde{\boldsymbol{v}}} \
\frac{\beta}{2}\|\sum_{j=1}^K \boldsymbol{A}_j\tilde{\boldsymbol{v}}_j
+\boldsymbol{b}-\boldsymbol{w}^{t-1}+
\frac{1}{\beta}\boldsymbol{\lambda}^{t-1}\|_2^2,
\\
\textstyle\boldsymbol{w}^{t}=\arg\min\limits_{\boldsymbol{w}}
\ f(\boldsymbol{w})+\frac{\beta}{2}\|\sum\limits_{j=1}^K
\boldsymbol{A}_j\tilde{\boldsymbol{v}}_j^{t}
+\boldsymbol{b}-\boldsymbol{w}+
\frac{1}{\beta}\boldsymbol{\lambda}^{t-1}\|_2^2,
\\
\textstyle\boldsymbol{\lambda}^{t}=\boldsymbol{\lambda}^{t-1}+
\beta(\sum_{j=1}^K \boldsymbol{A}_j\tilde{\boldsymbol{v}}_j^{t}
+\boldsymbol{b}-\boldsymbol{w}^{t}).
\end{array}
\right.
\label{ADMM-1}
\end{align}
All subproblems in (\ref{ADMM-1}) admit closed-form solutions.
We start with the $\tilde{\boldsymbol{v}}^{t}$-subproblem.
\subsubsection{Solving the $\tilde{\boldsymbol{v}}^{t}$-Subproblem}
Using the orthogonality property between $\boldsymbol{A}_j$ and
$\boldsymbol{A}_{j'}$, $\forall j\neq j'$, the
$\tilde{\boldsymbol{v}}^{t}$-subproblem can be decomposed as
\begin{align}
\tilde{\boldsymbol{v}}^{t}=\arg\min\limits_{\tilde{\boldsymbol{v}}} \
\textstyle\sum_{i=1}^Kg_i(\tilde{\boldsymbol{v}}_i;
\boldsymbol{w}^{t-1},\boldsymbol{\lambda}^{t-1}),
\label{ADMM-2}
\end{align}
where
\begin{align}
\textstyle g_i(\tilde{\boldsymbol{v}}_i;\boldsymbol{w}^{t-1},
\boldsymbol{\lambda}^{t-1})\triangleq\frac{\beta}{2}
\|\boldsymbol{A}_i\tilde{\boldsymbol{v}}_i
+\boldsymbol{D}_i(\boldsymbol{b}-\boldsymbol{w}^{t-1}+
\frac{1}{\beta}\boldsymbol{\lambda}^{t-1})\|_2^2,
\end{align}
and $\boldsymbol{D}_i$ is a $0/1$ binary diagonal matrix. The
$l$th diagonal element of $\boldsymbol{D}_i$ is $1$ only when
the $l$th row of $\boldsymbol{A}_i$ is nonzero. The solution
to the $\tilde{\boldsymbol{v}}^{t}$-subproblem is thus given as:
\begin{align}
\textstyle\tilde{\boldsymbol{v}}_i^t=
-(\boldsymbol{A}_i^T\boldsymbol{A}_i)^{-1}\boldsymbol{A}_i^T
\boldsymbol{D}_i(\boldsymbol{b}-\boldsymbol{w}^{t-1}+
\frac{1}{\beta}\boldsymbol{\lambda}^{t-1}), \ \forall i.
\label{ADMM-3-1}
\end{align}

\subsubsection{Solving the $\boldsymbol{w}^{t}$-Subproblem}
Recall that $\mathcal{D}\triangleq\mathcal{C}\times \cdots \times
\mathcal{C} \times \mathcal{P}_1 \times \cdots\mathcal{P}_M$.
Thus we have
$f(\boldsymbol{w})=\sum_{i=1}^{K+1} f_i(\boldsymbol{w}_i)$, where
\begin{align}
&\textstyle f_i(\boldsymbol{w}_i)
\triangleq\frac{1}{2}\|\boldsymbol{w}_i-
\text{Proj}_{\mathcal{C}}\{\boldsymbol{w}_i\}\|_2^2, \
1\leq i \leq K,
\nonumber\\
&\textstyle f_{K+1}(\boldsymbol{w}_{K+1})
\triangleq\frac{1}{2}\|\boldsymbol{w}_{K+1}
-\text{Proj}_{\mathcal{P}}\{\boldsymbol{w}_{K+1}\}\|_2^2.
\label{ADMM-4-1}
\end{align}
and $\mathcal{P}\triangleq \mathcal{P}_1 \times \cdots \times
\mathcal{P}_M$. The closed-form expressions of
$\text{Proj}_{\mathcal{C}}\{\boldsymbol{w}_i\}$ and
$\text{Proj}_{\mathcal{P}}\{\boldsymbol{w}_{K+1}\}$ will be
introduced later. Using the separable structure of $f$, the
$\boldsymbol{w}^{t}$-subproblem can be decomposed as
\begin{align}
&\textstyle\boldsymbol{w}^{t}=\arg\min\limits_{\boldsymbol{w}}
\ \sum_{i=1}^{K+1}f_i(\boldsymbol{w}_i)+
\frac{\beta}{2}\|\boldsymbol{w}_i-\boldsymbol{d}_i^{t-1}\|_2^2,
\label{ADMM-4}
\end{align}
where $\boldsymbol{d}_k^{t-1}$, $k\leq K$, is a vector composed
of the $((k-1)(2K+2)+1)$th to $(k(2K+2))$th element of
$\boldsymbol{A}_i\tilde{\boldsymbol{v}}_i^{t}+\boldsymbol{b}+
\frac{1}{\beta}\boldsymbol{\lambda}^{t-1}$, and
$\boldsymbol{d}_{K+1}^{t-1}$ is composed of the $(K(2K+2)+1)$th
to $(K(2K+2)+2MNK)$th element of
$\boldsymbol{A}_i\tilde{\boldsymbol{v}}_i^{t}+\boldsymbol{b}
+\frac{1}{\beta}\boldsymbol{\lambda}^{t-1}$.

From (\ref{ADMM-4}), we know that solving (\ref{ADMM-4})
amounts to minimizing each $f_i(\boldsymbol{w}_i)+
\frac{\beta}{2}\|\boldsymbol{w}_i-\boldsymbol{d}_i^{t-1}\|_2^2$,
that is
\begin{align}
\textstyle\boldsymbol{w}_i^{t}=\arg\min\limits_{\boldsymbol{w}_i}
\ f_i(\boldsymbol{w}_i)+\frac{\beta}{2}\|\boldsymbol{w}_i-
\boldsymbol{d}_i^{t-1}\|_2^2,
\ \forall i.
\label{ADMM-6}
\end{align}
The solution to (\ref{ADMM-6}), also referred to as the
proximal mapping of $f_i(\boldsymbol{w}_i)$, is given in
the following proposition.

\begin{proposition}
\label{proposition-2}
Consider the problem (\ref{ADMM-6}). Let
$\mathcal{C}_i\triangleq\mathcal{C}$, $1\leq i\leq K$, and
also let $\mathcal{C}_{K+1}\triangleq\mathcal{P}_1 \times
\cdots \times\mathcal{P}_M$, $i=K+1$. Then the
solution to (\ref{ADMM-6}) is given as
\begin{align}
\boldsymbol{w}_i^{t}=\left \{
\begin{array}{ll}
\boldsymbol{d}_i^{t-1},
& \text{if} \ \boldsymbol{d}_i^{t-1}\in\mathcal{C}_i,\\
\frac{\beta\boldsymbol{d}_i^{t-1}}{1+\beta}+
\frac{\emph{\text{Proj}}_{\mathcal{C}_i}\{\boldsymbol{d}_i^{t-1}\}}{1+\beta},
& \text{otherwise}.
\end{array}
\right.
\label{theorem2-2}
\end{align}
\end{proposition}
\begin{proof}
See Appendix \ref{appendix-A-1}.
\end{proof}

\subsubsection{Solutions to the Projections}
First consider $\text{Proj}_{\mathcal{C}_i}\{\boldsymbol{d}_i^{t-1}\}$,
$i\leq K$. Note that $\mathcal{C}_i$ is an SOC. Let
$[d_i^{t-1}]_{\text{last}}$ denote the last element of
$\boldsymbol{d}_i^{t-1}$ and also let
$[\boldsymbol{d}_i^{t-1}]_{\text{rest}}$ denote the subvector
containing the remaining elements. Then the closed-form solution
to $\text{Proj}_{\mathcal{C}_i}\{\boldsymbol{d}_i^{t-1}\}$,
$i\leq K$, is given as (see \cite{FukushimaLuo02})
\begin{align}
\text{Proj}_{\mathcal{C}_i}\{\boldsymbol{d}_i^{t-1}\} =
\left\{\begin{array}{ll}
\boldsymbol{d}_i^{t-1}, &\ \text{if} \
\|[\boldsymbol{d}_i^{t-1}]_{\text{rest}}\|_2
\leq [d_i^{t-1}]_{\text{last}}, \\
\boldsymbol{0},  &\ \text{if} \
\|[\boldsymbol{d}_i^{t-1}]_{\text{rest}}\|_2
\leq -[d_i^{t-1}]_{\text{last}}, \\
\bar{\boldsymbol{d}_i},  &\ \text{otherwise},
\end{array} \right.
\label{projection-soc}
\end{align}
where
\begin{align}
&\textstyle [\bar{\boldsymbol{d}_i}]_{\text{rest}}\triangleq
\frac{1}{2}(1+\frac{[d_i^{t-1}]_{\text{last}}}
{\|[\boldsymbol{d}_i^{t-1}]_{\text{rest}}\|_2})
\cdot [\boldsymbol{d}_i^{t-1}]_{\text{rest}},
\nonumber\\
&\textstyle [\bar{\boldsymbol{d}_i}]_{\text{last}}=\frac{1}{2}
([d_i^{t-1}]_{\text{last}}+
\|[\boldsymbol{d}_i^{t-1}]_{\text{rest}}\|_2).
\end{align}
As for $\text{Proj}_{\mathcal{C}_{K+1}}\{\boldsymbol{d}_{K+1}^{t-1}\}$,
notice that this is a length-$2MKN$ vector which can be vertically
split into $M$ non-overlapping subvectors, each of which is of
length $2KN$. Then $m$th subvector of
$\text{Proj}_{\mathcal{C}_{K+1}}\{\boldsymbol{d}_{K+1}^{t-1}\}$
is denoted as
$[\text{Proj}_{\mathcal{C}_{K+1}}\{\boldsymbol{d}_{K+1}^{t-1}\}]_m$,
which is obtained through the following formula
\begin{align}
[\text{Proj}_{\mathcal{C}_{K+1}}\{\boldsymbol{d}_{K+1}^{t-1}\}]_m
=\left \{
\begin{array}{ll}
[\boldsymbol{d}_{K+1}^{t-1}]_m,
\ \text{if} \ \|[\boldsymbol{d}_{K+1}^{t-1}]_m\|_2\leq p_m, \\
\frac{p_m\cdot[\boldsymbol{d}_{K+1}^{t-1}]_m}{\|[\boldsymbol{d}_{K+1}^{t-1}]_m\|_2},
\ \text{otherwise}.
\end{array} \right.
\end{align}

\subsection{Computational Analysis}
\label{sec-limitation}
The number of FLOPs (floating-point operations) for solving
the $\boldsymbol{w}^{t}$-subproblem is in the order of
$O(K^2+MNK)$. Here, $O(K^2)$ accounts for the $K$ projections
onto the SOC, and $O(MNK)$ corresponds to the projection
onto the power constraint set. The major complexity of the
standard ADMM lies in solving the
$\tilde{\boldsymbol{v}}^t$-subproblem, which mainly involves
two matrix-vector multiplications. The first one is the
multiplication between $\boldsymbol{A}_i^T$ and
$\boldsymbol{D}_i(\boldsymbol{b}-\boldsymbol{w}^{t-1}+
\frac{1}{\beta}\boldsymbol{\lambda}^{t-1})$ and the other
one is between $(\boldsymbol{A}_i^T\boldsymbol{A}_i)^{-1}$
and $\boldsymbol{A}_i^T\boldsymbol{D}_i(\boldsymbol{b}-
\boldsymbol{w}^{t-1}+\frac{1}{\beta}\boldsymbol{\lambda}^{t-1})$.
The number of FLOPs for computing the first multiplication
is in the order of $O(MNK)$ (notice the sparsity in
$\boldsymbol{A}_i^T$). As for the second multiplication,
first recall that
\begin{align}
\textstyle \boldsymbol{A}_i^T\boldsymbol{A}_i=\boldsymbol{I}+
\sum_{k=1}^{K}\tilde{\boldsymbol{H}}_k^T\tilde{\boldsymbol{H}}_k
+e(s_c)\cdot\tilde{\boldsymbol{h}}_{i}
\tilde{\boldsymbol{h}}_{i}^T, \ i\leq K.
\label{complexity-mul}
\end{align}
which means that (using the Woodbury formula)
\begin{align}
(\boldsymbol{A}_i^T\boldsymbol{A}_i)^{-1}=&
\boldsymbol{I}-\bar{\boldsymbol{A}}_i(\boldsymbol{I}+
\bar{\boldsymbol{A}}_i^H\bar{\boldsymbol{A}}_i)^{-1}\bar{\boldsymbol{A}}_i^H
\nonumber\\
=&\boldsymbol{I}-\underbrace{\bar{\boldsymbol{A}}_i(\boldsymbol{I}+
\bar{\boldsymbol{A}}_i^H\bar{\boldsymbol{A}}_i)^{-1/2}}_{\triangleq \boldsymbol{A}_{i,\text{left}}}
\underbrace{(\boldsymbol{I}+\bar{\boldsymbol{A}}_i^H\bar{\boldsymbol{A}}_i)^{-1/2}
\bar{\boldsymbol{A}}_i^H}_{\triangleq \boldsymbol{A}_{i,\text{left}}^T}
\nonumber\\
=&\boldsymbol{I}-
\boldsymbol{A}_{i,\text{left}}\cdot\boldsymbol{A}_{i,\text{left}}^T
\label{complexity-mul2}
\end{align}
where $\bar{\boldsymbol{A}}_i\triangleq [\tilde{\boldsymbol{H}}_1^T \
\cdots \ \tilde{\boldsymbol{H}}_K^T \
\sqrt{e(s_c)}\tilde{\boldsymbol{h}}_{i}]\in\mathbb{R}^{2MN\times
2K+1}$. Using (\ref{complexity-mul2}), we see
that multiplying the vector $\boldsymbol{A}_i^T\boldsymbol{D}_i(\boldsymbol{b}-
\boldsymbol{w}^{t-1}+\frac{1}{\beta}\boldsymbol{\lambda}^{t-1})$
by $(\boldsymbol{A}_i^T\boldsymbol{A}_i)^{-1}$ has a complexity of $O(MNK)$.
Therefore, the total number of FLOPs for solving the
$\tilde{\boldsymbol{v}}^{t}$-subproblem is in the order of
$O(MNK^2)$. At last, the number of FLOPs for updating
$\bar{\boldsymbol{\lambda}}^{t}$ is in the order
of $O(MNK^2)$. To summarize, the per-iteration complexity of
the standard ADMM is dominated by $O(MNK^2)$.

For large-scale CF massive MIMO systems with large values
of $M$, $N$, and $K$, its computational complexity could
become prohibitively high. This motivates us to develop a
randomized ADMM algorithm that solves the max-min beamforming
problem in a more efficient way.

\section{Proposed Randomized ADMM Algorithm}
\label{sec-proposed} In this section, we propose a randomized
ADMM algorithm for solving problem (\ref{problem-linear-cons2}).
The proposed algorithm is summarized in Algorithm \ref{alg:2}.
There are several differences between the proposed algorithm
and the standard ADMM (\ref{ADMM-1}). Firstly, in the
$\tilde{\boldsymbol{v}}^{t}$-subproblem, only a subset of
$\tilde{\boldsymbol{v}}_i$s are updated. The update formula
of $\tilde{\boldsymbol{v}}_i^t$ is given in (\ref{ADMM-3-1}).
The number of $\tilde{\boldsymbol{v}}_i$s to be updated is
determined by a selection probability $\alpha$. As discussed
earlier, the main computational cost of the standard ADMM
lies in the $\tilde{\boldsymbol{v}}^{t}$-subproblem. In our
algorithm, updating a smaller number of
$\tilde{\boldsymbol{v}}_i$s can significantly reduce the
per-iteration computational cost. Empirically, setting a
small value of $\alpha$ (say $\alpha=0.05$) is good enough
to attain a decent convergence speed. As such, the proposed
randomized ADMM is computationally more efficient than the
standard ADMM. Due to the random selection, each
$\tilde{\boldsymbol{v}}_i$ is updated every $1/\alpha$
iterations (in expectation).

To maintain a balance between
the primal and the dual updates, the dual update is also
modified. In fact, merging (\ref{sec2-2-1}) and (\ref{sec2-2})
yields
\begin{align}
&\textstyle\boldsymbol{\lambda}^{t}=\boldsymbol{\lambda}^{t-1}+
\alpha\beta(\sum_{j=1}^K \boldsymbol{A}_j\tilde{\boldsymbol{v}}_j^{t}
+\boldsymbol{b}-\boldsymbol{w}^{t}),
\label{our-alg-1-1}
\end{align}
which is $\alpha$ times slower than that in the standard
ADMM.

The $\boldsymbol{w}^{t}$-subproblem is similar to that in
the standard ADMM, except that: (1) the quadratic term
$\|\sum_{j=1}^K\boldsymbol{A}_j\tilde{\boldsymbol{v}}_j^{t}+
\boldsymbol{b}-\boldsymbol{w}\|_2^2$ is multiplied by the
selection probability $\alpha$; (2) an extra proximal term
$\frac{\bar{\alpha}\beta}{2}\|\boldsymbol{w}_i-
\boldsymbol{w}_i^{t-1}\|_2^2$ is added to this subproblem.
Intuitively, the first change is made to accommodate the
modified update of $\boldsymbol{\lambda}^{t}$, i.e.,
(\ref{our-alg-1-1}). While the second change is made to
ensure the convergence of the algorithm. To understand this,
recall that the extra proximal term forces the solution
to the $\boldsymbol{w}^{t}$-subproblem to stay near to
$\boldsymbol{w}^{t-1}$. When $\alpha$ is small,
$\boldsymbol{w}^{t}$ should proceed more cautiously
(because only a portion of $\tilde{\boldsymbol{v}}_i$s are
updated). Hence an extra proximal term is added to the
$\boldsymbol{w}^{t}$-subproblem. Theoretically, the parameter
$\bar{\alpha}$ should be chosen such that
$\alpha\bar{\alpha}\geq(\alpha^{-2}-1)$, see
Theorem \ref{theorem-content-2}. However, Theorem
\ref{theorem-content-2} imposes an overly pessimistic
constraint of $\bar{\alpha}$. In practical implementations,
$\bar{\alpha}$ does not need to be set this large.

\begin{algorithm}
\caption{Proposed R-ADMM} \label{alg:2}
\begin{algorithmic}
\STATE{\textbf{Inputs}: the selection probability $\alpha$,
the algorithm parameter $\beta$ and $\bar{\alpha}$, and the
maximum number of iterations $\bar{t}$. All initial vectors
are set to $\boldsymbol{0}$ (optional)}.
\STATE{\textbf{While} $t\leq \bar{t}$
\textbf{do}}
\STATE{\textcircled{1} \ \textbf{Selection}:
Each $g_i$ (see (\ref{ADMM-2})) has a probability of $\alpha$
to be selected. In the $t$th iteration, the index set of the
selected $g_i$s is denoted as $\Lambda^{t}$.}
\STATE{\textcircled{2} \ \textbf{Solving the}
$\tilde{\boldsymbol{v}}^{t}$-\textbf{subproblem}: }
\begin{align}
&\left\{\begin{array}{ll}
\tilde{\boldsymbol{v}}_i^{t}=
\arg\min\limits_{\tilde{\boldsymbol{v}}_i} \
\textstyle g_i(\tilde{\boldsymbol{v}}_i;\boldsymbol{w}^{t-1},
\boldsymbol{\lambda}^{t-1}), \ i\in \Lambda^{t},
\\
\tilde{\boldsymbol{v}}_i^{t}=\tilde{\boldsymbol{v}}_i^{t-1},
\ i\notin \Lambda^{t}.
\end{array} \right.
\label{alg-2}
\end{align}
where $g_i(\tilde{\boldsymbol{v}}_i;\boldsymbol{w}^{t-1},
\boldsymbol{\lambda}^{t-1})\triangleq\frac{\beta}{2}
\|\boldsymbol{A}_i\tilde{\boldsymbol{v}}_i
+\boldsymbol{D}_i(\boldsymbol{b}-\boldsymbol{w}^{t-1}+
\frac{1}{\beta}\boldsymbol{\lambda}^{t-1})\|_2^2$ and the
solution to this subproblem is given in (\ref{ADMM-3-1}).
\STATE{\textcircled{3} \ \textbf{Solving the}
$\boldsymbol{w}^t$-\textbf{subproblem}:
\begin{align}
&\textstyle\boldsymbol{w}^{t}=\arg\min\limits_{\boldsymbol{w}}
f(\boldsymbol{w})+\langle\boldsymbol{\lambda}^{t-1},\sum_{j=1}^K
\boldsymbol{A}_j\tilde{\boldsymbol{v}}_j^{t}+\boldsymbol{b}-
\boldsymbol{w}\rangle+
\nonumber\\
&\textstyle\frac{\alpha\beta}{2}\|\sum_{j=1}^K
\boldsymbol{A}_j\tilde{\boldsymbol{v}}_j^{t}+\boldsymbol{b}-
\boldsymbol{w}\|_2^2
+\frac{\bar{\alpha}\beta}{2}
\|\boldsymbol{w}-\boldsymbol{w}^{t-1}\|_2^2, \ t<\bar{t},
\label{alg-1}
\\
&\textstyle\boldsymbol{w}^{\bar{t}}=
\arg\min\limits_{\boldsymbol{w}}
\ f(\boldsymbol{w})+\langle\boldsymbol{\lambda}^{t-1},
\sum_{j=1}^K\boldsymbol{A}_j\tilde{\boldsymbol{v}}_j^{\bar{t}}+
\boldsymbol{b}-\boldsymbol{w}\rangle+
\nonumber\\
&\textstyle\frac{\beta}{2}\|\sum_{j=1}^K\boldsymbol{A}_j
\tilde{\boldsymbol{v}}_j^{\bar{t}}+
\boldsymbol{b}-\boldsymbol{w}\|_2^2+
\frac{\hat{\alpha}\beta}{2}
\|\boldsymbol{w}-\boldsymbol{w}^{\bar{t}-1}\|_2^2,
\label{alg-1-1}
\end{align}
where $\hat{\alpha}\triangleq\alpha\bar{\alpha}$.
}
\STATE{\textcircled{4} \ \textbf{Updating the Lagrangian multiplier}: }
\begin{align}
&\textstyle\bar{\boldsymbol{\lambda}}^{t}=
\boldsymbol{\lambda}^{t-1}+
\beta(\sum_{j=1}^K \boldsymbol{A}_j\tilde{\boldsymbol{v}}_j^{t}
+\boldsymbol{b}-\boldsymbol{w}^{t}),
\label{sec2-2-1}
\\
&\boldsymbol{\lambda}^{t}=\boldsymbol{\lambda}^{t-1}+
\alpha(\bar{\boldsymbol{\lambda}}^{t}-\boldsymbol{\lambda}^{t-1}).
\label{sec2-2}
\end{align}
\par{\textbf{End While and Output}
$\tilde{\boldsymbol{v}}^{\bar{t}}$;}
%\par{\textbf{Outputs: $\tilde{\boldsymbol{v}}^{\bar{t}}$};}
\end{algorithmic}
\end{algorithm}

\subsection{Efficient Implementations and Computational Complexity}
\label{sec-proposed-sub1}
\subsubsection{$\boldsymbol{w}^{t}$-subproblem}
For this subproblem, first rewrite (\ref{alg-1}) as
\begin{align}
&\textstyle\boldsymbol{w}^{t}=
\nonumber\\
&\textstyle\arg\min\limits_{\boldsymbol{w}}
\ f(\boldsymbol{w})+\frac{(\alpha+\bar{\alpha})\beta}{2}
\|\boldsymbol{w}-(\frac{\alpha}{\alpha+
\bar{\alpha}}\boldsymbol{d}^{t-1}+
\frac{\bar{\alpha}}{\alpha+\bar{\alpha}}\boldsymbol{w}^{t-1})\|_2^2,
\nonumber\\
&\qquad\qquad \qquad t< \bar{t}-1,
\label{our-alg-1}
\end{align}
where $\boldsymbol{d}^{t-1}\triangleq\sum_{j=1}^K
\boldsymbol{A}_j\tilde{\boldsymbol{v}}_j^{t}+\boldsymbol{b}
+\frac{1}{\alpha\beta}\boldsymbol{\lambda}^{t-1},
\ t< \bar{t}-1$. Further using the separable structure
of $f$, we know that solving (\ref{our-alg-1}) amounts to
solving
\begin{align}
&\textstyle\boldsymbol{w}_i^{t}=
\nonumber\\
&\textstyle
\arg\min\limits_{\boldsymbol{w}_i}
\ f_i(\boldsymbol{w}_i)+\frac{(\alpha+\bar{\alpha})\beta}{2}
\|\boldsymbol{w}_i-(\frac{\alpha}{\alpha+\bar{\alpha}}
\boldsymbol{d}_i^{t-1}+\frac{\bar{\alpha}}{\alpha+
\bar{\alpha}}\boldsymbol{w}_i^{t-1})\|_2^2,
\nonumber\\
&\qquad\qquad \qquad 1\leq i\leq K+1,
\label{our-alg-2}
\end{align}
where $f_i(\boldsymbol{w}_i)$ is defined in (\ref{ADMM-4-1}).
The solution to (\ref{our-alg-2}) can be obtained from
Proposition \ref{proposition-2}. In the last iteration, the
$\boldsymbol{w}^{\bar{t}}$-update is changed to (\ref{alg-1-1}).
Such a modification is only for ease of analysis. In fact,
there is no need to perform (\ref{alg-1-1}) because the
output of Algorithm \ref{alg:2} is $\tilde{\boldsymbol{v}}^{\bar{t}}$,
which is already obtained before the
$\boldsymbol{w}^{\bar{t}}$-subproblem.

\subsubsection{$\tilde{\boldsymbol{v}}^{t}$-subproblem}
The update of $\tilde{\boldsymbol{v}}_i^{t}$, $1\leq i\leq K$,
is given in (\ref{ADMM-3-1}). From
(\ref{complexity-mul}), we know that
if we have already obtained $(\boldsymbol{I}+\sum_{k=1}^{K}
\tilde{\boldsymbol{H}}_k^T\tilde{\boldsymbol{H}}_k)^{-1}$,
the inverse of $\boldsymbol{A}_i^T\boldsymbol{A}_i$ can be
conveniently computed using the Woodbury
formula. Notably, the matrix $(\boldsymbol{I}+\sum_{k=1}^{K}
\tilde{\boldsymbol{H}}_k^T\tilde{\boldsymbol{H}}_k)^{-1}$
needs to be computed only once.

\subsubsection{$\bar{\boldsymbol{\lambda}}^{t}$-subproblem}
Note that the update of $\bar{\boldsymbol{\lambda}}^{t}$
involves calculating $\sum_{j=1}^K
\boldsymbol{A}_j\tilde{\boldsymbol{v}}_j^{t}$. In fact,
there is no need to compute every $\boldsymbol{A}_j
\tilde{\boldsymbol{v}}_j^{t}$ and then sum them together
within each iteration. To see this, suppose we maintain a
vector $\boldsymbol{\tau}^{t-1}=\sum_{j=1}^K
\boldsymbol{A}_j\tilde{\boldsymbol{v}}_j^{t-1}$. Recall
that in the $t$th iteration, only those
$\tilde{\boldsymbol{v}}_i^{t}$s, $i\in\Lambda^{t}$, are
updated. Thus, we have $\boldsymbol{\tau}^{t}=\textstyle
\boldsymbol{\tau}^{t-1}+\sum_{i\in\Lambda^{t}}
\boldsymbol{A}_i(\tilde{\boldsymbol{v}}_i^{t}-
\tilde{\boldsymbol{v}}_i^{t-1})$. As a result, only
$|\Lambda^{t}|$ matrix-vector multiplications need to be
performed.

\subsubsection{Computational Complexity}
\label{sec-proposed-complexity} Similar to the discussions
in Section \ref{sec-limitation}, it can be easily verified
that the per-iteration computational complexity of
Algorithm \ref{alg:2} is in the order of $O(\alpha MNK^2)$,
which is $\alpha\ll 1$ times that of the standard ADMM.

\section{Convergence Analysis}
\label{sec-convergence}
In this section, we prove the sublinear convergence rate
for the proposed randomized ADMM algorithm. Our main
results are summarized in the following theorem.

\newtheorem{theorem}{Theorem}
\begin{theorem}
\label{theorem-content-2} Let
$\{\boldsymbol{w}^{*},\tilde{\boldsymbol{v}}^*,
\boldsymbol{\lambda}^*\}$ denote a set of optimal primal-dual
solution to the linearly constrained problem
(\ref{problem-linear-cons2}). Suppose in Algorithm \ref{alg:2},
the selection probability is set to $\alpha$, and the maximum
number of iterations is set to $\bar{k}$. Furthermore, suppose
the parameter $\bar{\alpha}$ is chosen such that
\begin{align}
\alpha\bar{\alpha}\geq (\alpha^{-2}-1).
\label{theorem-1}
\end{align}
Then the sequence generated by Algorithm \ref{alg:2} satisfies
\begin{align}
&\textstyle \mathbb{E}\big[
|f(\vec{\boldsymbol{w}}^{\bar{t}})-f(\boldsymbol{w}^{*})|\big]
\leq \frac{\emph{\text{Const}}+\frac{1}{2\beta}\|\boldsymbol{\lambda}^{0}-
\vec{\boldsymbol{\lambda}}\|_2^2}{1+\alpha\cdot(\bar{t}-1)}
\label{theorem-2}
\end{align}
and
\begin{align}
&\textstyle \mathbb{E}\big[
\|\boldsymbol{A}\vec{\boldsymbol{v}}^{\bar{t}}
+\boldsymbol{b}-\vec{\boldsymbol{w}}^{\bar{t}}\|_2\big]
\leq \frac{\emph{\text{Const}}+\frac{1}{2\beta}\|\boldsymbol{\lambda}^{0}-
\vec{\boldsymbol{\lambda}}\|_2^2}{C\cdot(1+\alpha\cdot(\bar{t}-1))}
\label{theorem-3}
\end{align}
where the expectation is taken over all possible realizations due
to random selection of $\boldsymbol{\tilde{v}}^{t}$-subproblems,
$\boldsymbol{\lambda}^{0}$ is the initial vector of
$\boldsymbol{\lambda}$, \emph{Const} is a constant number, and
\begin{align}
&\textstyle\vec{\boldsymbol{w}}^{\bar{t}}
\triangleq\frac{\boldsymbol{w}^{\bar{t}}+\alpha\sum_{t=1}^{\bar{t}-1}
\boldsymbol{w}^{t}}{1+\alpha\cdot(\bar{t}-1)}, \
\vec{\boldsymbol{v}}^{\bar{t}}\triangleq
\frac{\tilde{\boldsymbol{v}}^{\bar{t}}+
\alpha\sum_{t=1}^{\bar{t}-1}
\tilde{\boldsymbol{v}}^{t}}{1+\alpha\cdot(\bar{t}-1)},
\nonumber\\
&\textstyle\vec{\boldsymbol{\lambda}}=2C\cdot
\frac{\boldsymbol{A}\vec{\boldsymbol{v}}^{\bar{t}}+
\boldsymbol{b}-\vec{\boldsymbol{w}}^{\bar{t}}}
{\|\boldsymbol{A}\vec{\boldsymbol{v}}^{\bar{t}}
+\boldsymbol{b}-\vec{\boldsymbol{w}}^{\bar{t}}\|_2}, \
C\triangleq\|\boldsymbol{\lambda}^*\|_2+\epsilon,
\end{align}
in which $\epsilon$ is a small positive scalar.
\end{theorem}
\begin{proof}
See Appendix \ref{appendix-B}.
\end{proof}
\newtheorem{comment}{Comment}
\begin{comment}
Since (\ref{theorem-2}) and (\ref{theorem-3}) hold simultaneously,
$\{\vec{\boldsymbol{w}}^{t},\vec{\boldsymbol{v}}^{t}\}$ is guaranteed
to converge to the optimal solution of (\ref{problem-linear-cons2}).
\end{comment}
\begin{comment}
\label{comment-2}
Although (\ref{theorem-2}) (resp. (\ref{theorem-3})) is established
w.r.t. the weighted average of all past variables, i.e.,
$\vec{\boldsymbol{w}}^{\bar{t}}$ (resp. $\vec{\boldsymbol{v}}^{\bar{t}}$),
there is no need to compute them in practice
because such a time average is overly pessimistic. Instead, the
instantaneous output $\{\tilde{\boldsymbol{v}}^{\bar{t}},
\boldsymbol{w}^{\bar{t}}\}$ should be the output of the algorithm.
\end{comment}

\section{Extension To QoS-Aware Beamforming}
\label{sec-extension-QoS}
Consider the following QoS-aware beamforming problem:
\begin{align}
\textstyle
\mathop {\min }\limits_{\{\boldsymbol{v}_{k}\}_{1\leq k\leq K}} &
\textstyle \ \sum_{m=1}^M\sum_{k=1}^K ||\boldsymbol{v}_k[m]||_2^2
\nonumber\\
\text{s.t.} & \textstyle \
R_k(\{\boldsymbol{v}_{k'}\}_{1\leq k'\leq K})\geq s_c, \
\forall k,
\label{QoS-cons-beam}
\end{align}
where $s_c$ is a pre-specified parameter. Problem
(\ref{QoS-cons-beam}) aims to find a set of minimum-cost
beamforming vectors that fulfill the QoS requirements of all
users. Following the derivations in Section \ref{sec-linear},
we can see that the QoS-aware beamforming problem
(\ref{QoS-cons-beam}) admits the following formulation:
\begin{align}
\mathop {\min }\limits_{\{\tilde{\boldsymbol{v}}_i\}_{1\leq i\leq K}} &
\textstyle \ \sum_{i=1}^K \|\tilde{\boldsymbol{v}}_i\|_2^2
\nonumber\\
\text{s.t.} & \textstyle \
\boldsymbol{G}\tilde{\boldsymbol{v}}+
\boldsymbol{e}\in\mathcal{D}_{pm},
\label{our-alg-extension}
\end{align}
where $\boldsymbol{G}$ and $\boldsymbol{e}$ are constructed
similarly as $\boldsymbol{H}$ and $\boldsymbol{b}$ in
(\ref{feasibility-compact-1-1}), and $\mathcal{D}_{pm}
\triangleq \mathcal{C}\times\cdots\times \mathcal{C}$ is
the Cartesian product of $K$ $\mathcal{C}$s.
Introducing an auxiliary variable $\boldsymbol{w}$ to
(\ref{our-alg-extension}) yields the following problem:
\begin{align}
\mathop {\min }\limits_{\{\tilde{\boldsymbol{v}}_i\}_{1\leq i\leq K},
\ \boldsymbol{w}} &
\textstyle \ \sum_{i=1}^K \|\tilde{\boldsymbol{v}}_i\|_2^2
\nonumber\\
\text{s.t.} & \textstyle \ \boldsymbol{w}\in\mathcal{D}_{pm},
\nonumber\\
&\textstyle
\ \boldsymbol{G}\tilde{\boldsymbol{v}}+\boldsymbol{e}
=\boldsymbol{w}.
\label{our-alg-extension-2}
\end{align}
Applying ADMM for (\ref{our-alg-extension-2}) yields
\begin{align}
&\tilde{\boldsymbol{v}}^{t}=
\arg\min\limits_{\tilde{\boldsymbol{v}}} \
\textstyle(\sum_{i=1}^K
\|\tilde{\boldsymbol{v}}_i\|_2^2)+
\frac{\beta}{2}\|\boldsymbol{G}\tilde{\boldsymbol{v}}+
\boldsymbol{e}-\boldsymbol{w}^{t-1}
%\nonumber\\
%%&\textstyle
%%\qquad\qquad\qquad\qquad\qquad\qquad\qquad
+\frac{1}{\sigma}\boldsymbol{\lambda}^{t-1}\|_2^2,
\nonumber\\
&\textstyle\boldsymbol{w}^{t}
=\arg\min\limits_{\boldsymbol{w}\in\mathcal{D}_{pm}}
\
\frac{\beta}{2}\|\boldsymbol{G}\tilde{\boldsymbol{v}}^{t}+
\boldsymbol{e}-\boldsymbol{w}+
\frac{1}{\beta}\boldsymbol{\lambda}^{t-1}\|_2^2,
\nonumber\\
&\textstyle\boldsymbol{\lambda}^{t}=\boldsymbol{\lambda}^{t-1}+
\beta(\boldsymbol{G}\tilde{\boldsymbol{v}}^{t}+
\boldsymbol{e}-\boldsymbol{w}^{t}).
\label{our-alg-extension-3}
\end{align}
Using the orthogonal structure inherent in $\boldsymbol{G}$,
the $\tilde{\boldsymbol{v}}^{t}$-subproblem can be decomposed
into $K$ independent subtasks, each of which only involves
$\tilde{\boldsymbol{v}}_i$. Moreover, the
$\boldsymbol{w}^{t}$-subproblem admits a closed-form
solution, that is, $\boldsymbol{w}^{t}=
\text{Proj}_{\mathcal{D}_{pm}}\{\boldsymbol{G}
\tilde{\boldsymbol{v}}^{t}+\boldsymbol{e}+
\frac{1}{\beta}\boldsymbol{\lambda}^{t-1}\}$.
For this reason, (\ref{our-alg-extension-3}) can be
similarly solved by the proposed Algorithm \ref{alg:2}.

\section{Relation to Existing Works}
\label{sec-relation}
\subsection{Relation to randomized ADMM}
Previous works such as \cite{SunLuo20,HongChang20,MihicZhu21}
have also developed randomized ADMM algorithms. However,
our proposed Algorithm \ref{alg:2} is different
from existing methods. First of all, methods in
\cite{SunLuo20,HongChang20,MihicZhu21} are developed for
problems of the following form:
\begin{align}
\textstyle\mathop {\min}\limits_{\tilde{\boldsymbol{x}}}  \ &
\textstyle
f(\tilde{\boldsymbol{x}})\triangleq g(\boldsymbol{x}_1,
\cdots,\boldsymbol{x}_K)+\sum_{k=1}^K h_k(\boldsymbol{x}_k),
\nonumber\\
\text{s.t.} \ & \boldsymbol{E}_{1}\boldsymbol{x}_1+\cdots+
\boldsymbol{E}_{K}\boldsymbol{x}_K=\boldsymbol{q}
\label{relation-existing-1}
\end{align}
where $\tilde{\boldsymbol{x}}\triangleq[\boldsymbol{x}_1;
\cdots;\boldsymbol{x}_K]\in\mathbb{R}^{KN}$, $\boldsymbol{x}_k
\in\mathbb{R}^{N}$, $1\leq k\leq K$, $\boldsymbol{E}_{k}\in
\mathbb{R}^{M\times N}$, $1\leq k\leq K$, is the given data
matrix, and $\boldsymbol{q}\in\mathbb{R}^M$ is a given data
vector. If we treat $\boldsymbol{w}$ and each $\tilde{\boldsymbol{v}}_j$
as independent blocks, then problem (\ref{problem-linear-cons2})
also falls into the form of (\ref{relation-existing-1}). The
major difference between our algorithm and those in
\cite{SunLuo20,HongChang20,MihicZhu21} is how
these block variables are updated. For example, the method
in \cite{HongChang20} randomly chooses a block variable
to update in each iteration, including both primal variables
(e.g. $\boldsymbol{w}$ and each $\tilde{\boldsymbol{v}}_j$) and
dual variables (e.g. $\lambda$). For our proposed algorithm,
the $\boldsymbol{w}$ and $\lambda$ variables are updated in
every iteration, while random selection is only applied to
the sub-blocks of the $\tilde{\boldsymbol{v}}$ variable. Our
design is beneficial for accelerating the convergence because
it reduces the randomness in the algorithm. Since our method
is essentially different from existing methods, the proof
techniques used in existing literatures can not be
straightforwardly applied to analyzing our algorithm.

\subsection{Relation to Uplink-Downlink Duality-based Methods}
\label{sec-relation-udd}
The UDD-based (uplink-downlink duality) approach
\cite{SchubertBoche04,HuangTan13,GongJordan09,YuLan07,
BjornsonJorswieck13,MirettiCavalcante24,MirettiCavalcante23,Nuzman07}
originates in a completely different perspective. By leveraging
the Lagrangian duality, the UDD-based approach transforms
the highly challenging downlink beamforming problem into an
equivalent uplink dual problem which has a favorable
separable structure. Due to this, the uplink dual problem
can be efficiently solved via alternating optimization.
However, when dealing with more involved settings such
as multi-cell networks or cell-free networks with per-AP
or per-antenna power constraints, the uplink dual problem
becomes a highly coupled SDP which also calls for advanced
optimization methods. Take \cite{YuLan07} as an example.
This work provides an alternative solution to the
feasibility-checking problem (\ref{feasibility-problem}).
Therefore the method in \cite{YuLan07} can be used to
solve the MMB problem (\ref{maxmin-beam}) when combined with the
bisection search. Specifically, \cite{YuLan07} aims to
solve the following problem:
\begin{align}
\textbf{\text{Downlink-primal}}: \
& \mathop {\min }_{\{\boldsymbol{v}_k[m]\}_{k,m},
\ \alpha} \ \alpha \nonumber\\
\text{s.t.} & \textstyle \
\sum_{k=1}^K ||\boldsymbol{v}_k[m]||_2^2\leq \alpha \cdot p_m,
\ \forall m,
\nonumber\\
& \textstyle \
R_k(\{\boldsymbol{v}_{k'}\}_{k'=1}^K)\geq s_c
\label{ref1-downlink-perAP}
\end{align}
where $\alpha$ is an optimization variable. Note that
problem (\ref{ref1-downlink-perAP}) is similar to problem
(\ref{feasibility-problem}), except that (\ref{ref1-downlink-perAP})
imposes a ``soft'' per-AP power constraint by introducing
the variable $\alpha$. To explain the equivalence between
problem (\ref{ref1-downlink-perAP}) and (\ref{feasibility-problem}),
let the solution to (\ref{ref1-downlink-perAP}) be denoted
as $\{\alpha^*,\{\boldsymbol{v}^*_k[m]\}_{k,m}\}$. Clearly,
if $\alpha^*\leq 1$, then $\{\boldsymbol{v}^*_k[m]\}_{k,m}$
is also a feasible solution to problem (\ref{feasibility-problem}).
On the other hand, if $\alpha^*\geq 1$, then problem
(\ref{feasibility-problem}) dose not have a feasible
solution. Therefore the solution to (\ref{ref1-downlink-perAP})
can help determine the feasibility of the problem
(\ref{feasibility-problem}) in our paper. To solve
(\ref{ref1-downlink-perAP}), references [13-15] proposed
to convert it into an uplink dual problem given as:
\begin{align}
&\textbf{\text{Uplink-dual}}:
\nonumber\\
&\textstyle \mathop {\max }\limits_{\boldsymbol{Q}, \
\{\lambda_k\}_{k=1}^K}  \ \textstyle
\sum_{k=1}^K \lambda_k\cdot \sigma^2
\nonumber\\
\text{s.t.} & \ \textstyle
\boldsymbol{Q}+\sum\limits_{k'=1}^K\lambda_{k'}\cdot\boldsymbol{h}_{k'}\boldsymbol{h}_{k'}^H
\succeq \left(1+\frac{1}{s_{\text{fix}}}\right)
\lambda_i\cdot\boldsymbol{h}_{k}\boldsymbol{h}_{k}^H, \ 1\leq k \leq K,
\nonumber\\
& \text{tr}\{\boldsymbol{Q}\boldsymbol{\Phi}\}\leq \text{tr}\{\boldsymbol{\Phi}\},
\ \boldsymbol{Q} \ \text{diagonal}, \  \boldsymbol{Q} \succeq \boldsymbol{0}.
\label{ref1-uplink}
\end{align}
where $\{\lambda_k\}_{k=1}^K$ is a set of Lagrangian dual
variables, $\boldsymbol{Q}\triangleq
\text{diag}\{q_1,\cdots,q_1,q_2,\cdots,q_2,\cdots,q_M,\cdots,q_M\}
\in\mathbb{R}^{MN\times MN}$ is a variable diagonal matrix, and
$\boldsymbol{\Phi}\triangleq \frac{1}{N}
\cdot\text{diag}\{p_1,\cdots,p_1,p_2,\cdots,p_2,\cdots,
p_M,\cdots,p_M\}\in\mathbb{R}^{MN\times MN}$. After (\ref{ref1-uplink})
is solved, the solution to (\ref{ref1-downlink-perAP}) can
be accordingly determined. Clearly, problem (\ref{ref1-uplink})
is a highly coupled SDP since the variables
$\{\lambda_k\}_{k=1}^K$ and $\boldsymbol{Q}$ appear in all
the constraints. For this reason, solving the uplink problem
is no easier than solving the original downlink problem.

\section{Simulation Results}
\label{sec-simulation} In this section, we conduct numerical
experiments to demonstrate the superiority of the proposed R-ADMM
algorithm. We compare our algorithm with the standard ADMM (S-ADMM),
i.e., Algorithm (\ref{ADMM-1}), the DRS (also referred to
as the averaged alternating reflection algorithm (AARA))
\cite{BauschkeMoursi17,WangWang21} as well as the uplink-downlink
duality-based method (UDDm) \cite{YuLan07,BjornsonJorswieck13}.
Although HSDE \cite{ShiZhang15} is also a first-order method,
in the large-scale regime, it requires computing the inverse
of a data matrix that is too large to be inverted. As such,
HSDE is not included in the comparison. All the simulations are
carried out on a desktop with a 3.7 GHz quad-core Intel Core i9
processor and $128$ GB of RAM.

\subsection{Experimental Settings}
Firstly, we assume that the $M=16$ APs and $K$ users are
independently and uniformly distributed within a square
area of size $500$m$\times$$500$m. The spatially correlated
channel vector between the $m$th AP and the $k$th user, namely,
$\boldsymbol{h}_k[m]\in\mathbb{C}^{N}$, is generated as
\begin{align}
\boldsymbol{h}_k[m]=(\varrho_{k}[m])^{1/2}\boldsymbol{g}_k[m]
\end{align}
where $\varrho_{k}[m]\in\mathbb{R}$ denotes the large-scale
fading, and $\boldsymbol{g}_k[m]\in\mathbb{C}^{N}$ denotes
the small-scale fading vector between the $m$th AP and the
$k$th user. The large-scale fading coefficient $\varrho_{k}[m]$
is modeled (in dB) as \cite{OzdoganBjornson19}
\begin{align}
\varrho_{k}[m] = -34.53 - 38\text{log}_{10}(d_m[k]/1m)+\chi_{m}[k]
\end{align}
where $d_m[k]$ is the distance between the $m$th AP and the
$k$th user, and $\chi_{m}[k]\thicksim \mathcal{N}(0,100)$
represents shadow fading. We consider a spatially correlated
small-scale fading scenario where the small-scale fading
vector is randomly generated according to $\boldsymbol{g}_k[m]\thicksim
\mathcal{CN}(\boldsymbol{0},\boldsymbol{R}_k[m])$, where
$\boldsymbol{R}_k[m]\in\mathbb{C}^{N\times N}$ is the positive
semi-definite spatial correlation matrix generated according
to \cite{BjornsonHoydis17}. At last, the noise power $\sigma_k$
(in dBm) is determined via
\begin{align}
\sigma_k=-174+10\text{log}_{10}(\text{bandwidth})
\end{align}
where the bandwidth is set to $20$MHz in our simulations.

\begin{figure}[!htbp]
    \centering
    \includegraphics[height=4.5cm,width=6cm]{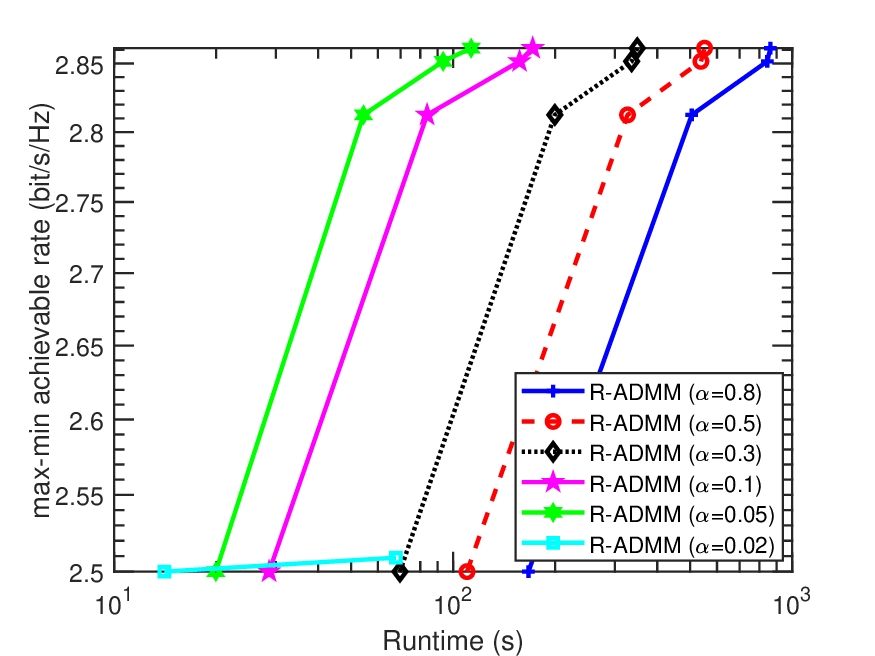}
    \caption{Performance of R-ADMM with different sampling rates}
    \label{fig1}
\end{figure}

\begin{figure}[!htbp]
    \centering
    \includegraphics[height=4.5cm,width=6cm]{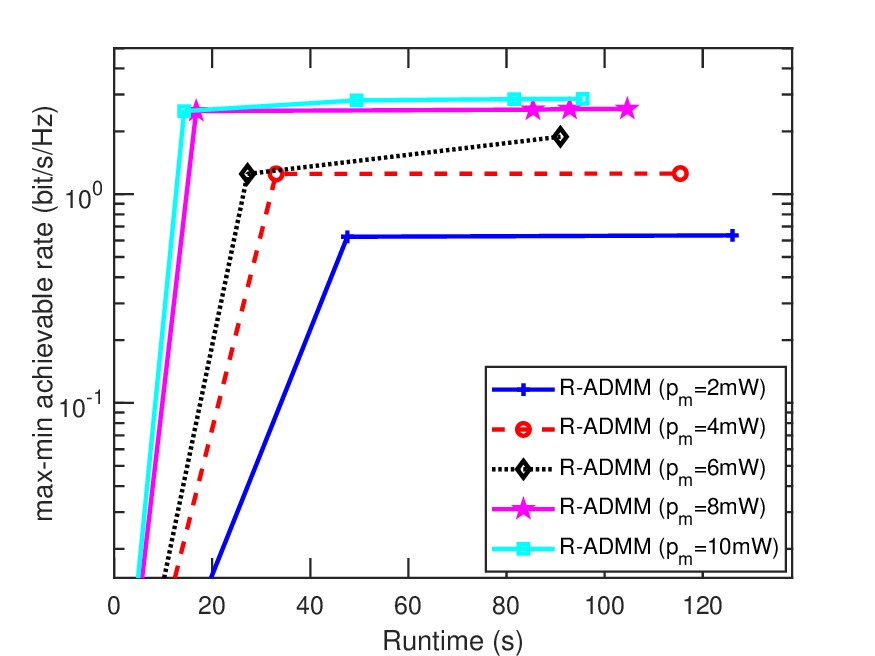}
    \caption{Performance of R-ADMM with different transmit powers}
    \label{fig2}
\end{figure}

To evaluate the performance of respective algorithms, we run
the bisection algorithm, namely, Algorithm \ref{alg:bsm},
with $[s_{\text{min}},s_{\text{max}}]$ initialized
as $[0,10]$ and $s_{\text{ter}}$ set to $0.01$. For such a
setting, the bisection search should be performed $10$ times to
arrive at the $0.01$ accuracy. When running the bisection
algorithm, it is critical to correctly identify the feasibility
or infeasibility of problem (\ref{feasibility-problem}).
To do so, the stopping criterion for each algorithm is set as
follows: If the optimality gap, defined as $\text{opg}^t\triangleq
\|\tilde{\boldsymbol{v}}^t -\tilde{\boldsymbol{v}}^{t-1}\|_2$,
decreases to $10^{-10}$, then the algorithm terminates; otherwise
each algorithm continues until it reaches the maximum number of
iterations, which is set $5000$ in our simulations. After the
algorithm stops, the feasibility or infeasibility of problem
(\ref{feasibility-problem}) is confirmed using the objective
function value attained by the algorithm. Finally, it should be
noted that the pre-processing task of all algorithms are the
same. Therefore the time consumption of pre-processing is not
included.

\begin{figure*}[!htbp]
    \centering
    \subfigure[]
    {\includegraphics[height=4.5cm,width=6cm]{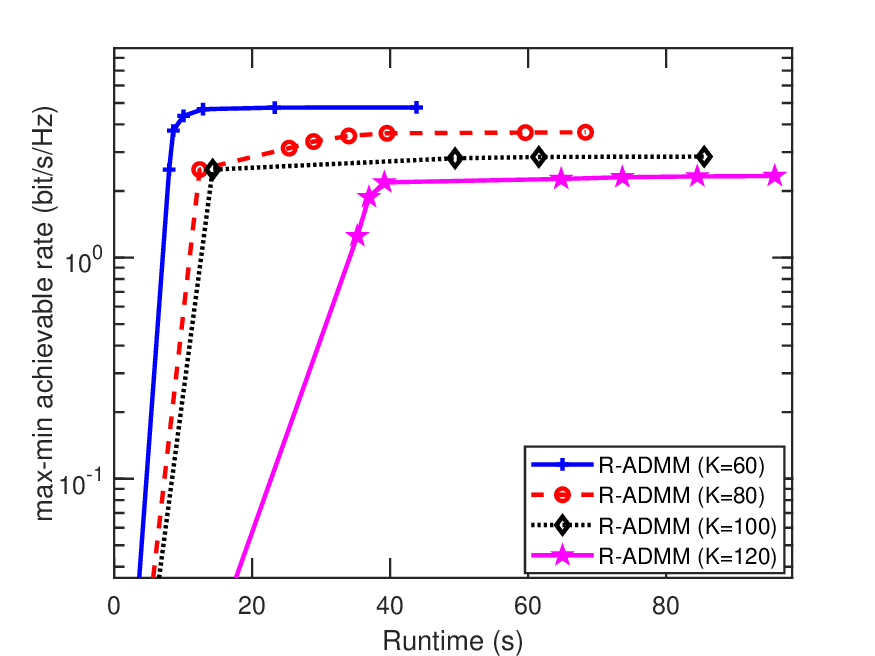}} \hfil
    \subfigure[]
    {\includegraphics[height=4.5cm,width=6cm]{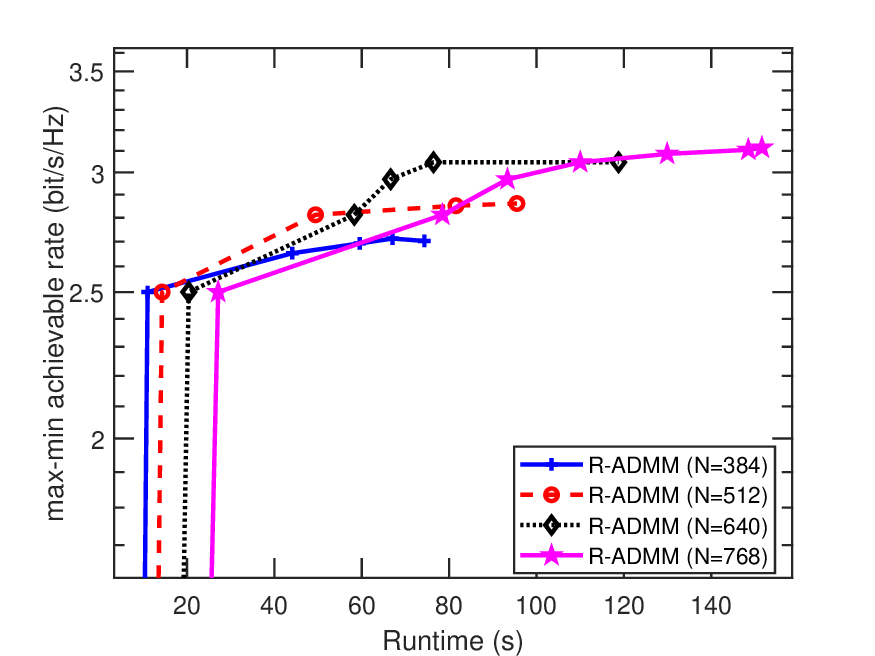}} \hfil
    \caption{Scalability of R-ADMM. Left: Performance of R-ADMM with
    different $K$; Right: Performance of R-ADMM with different $N$;}
    \label{fig3}
\end{figure*}

%Next, we take an example to explain how to generate the
%results in Fig. \ref{fig1}. At the first iteration of
%bisection, we have $s_c=(s_{\text{min}}+s_{\text{max}})/2$.
%Suppose problem (\ref{feasibility-problem}) with this
%value of $s_c$ is identified as feasible, then $5$ is
%recorded as an \emph{achievable max-min rate} (see the
%vertical axis of Fig. \ref{fig1}). In the next bisection
%iteration, we set $s_c=7.5$. Now suppose problem
%(\ref{feasibility-problem}) is identified as infeasible,
%then $7.5$ won't be recorded and will be skipped. In other
%words, only those feasible values of $s_c$ will be recorded
%as the achievable max-min rate. Meanwhile, suppose the
%$i$th achievable max-min rate is found in the $i'$th
%bisection iteration, then the runtime corresponding to
%this achievable max-min rate is the entire runtime of
%the first $i$ bisection iterations.

\subsection{Experimental Results}
First, we investigate the performance of the proposed R-ADMM
under different sampling rates. The number of antennas at each
AP is set to $36$, and the number of users is set to $100$.
Besides, the per-AP transmit power $p_m$, $1\leq m\leq M$,
is fixed as $10$mW. The parameter $\bar{\alpha}$ is fixed
as $0.01$. Fig. \ref{fig1} plots the convergence behavior
of the proposed algorithm as a function of the runtime,
where the sampling rate $\alpha$ varies from $0.02$ to $0.8$.
Note that the max-min achievable rate is a monotonically
increasing function of the runtime.

From Fig. \ref{fig1}, we see that the proposed R-ADMM under
different sampling rate choices, except for $\alpha=0.02$,
converges to the same achievable rate, which is around $2.86$bit/s/Hz.
When $\alpha=0.02$, the proposed algorithm converge to a smaller
achievable rate, namely, around $2.51$bit/s/Hz. The reason is that for
a very small value of $\alpha$, the R-ADMM algorithm would
require more than $5000$ iterations to converge. Since, in our
experiments, the maximum number of iterations is set to $5000$,
the R-ADMM may be terminated before its convergence, in which
case the feasible case could be incorrectly identified as an
infeasible case. Therefore it converges to a smaller achievable
rate. From Fig. \ref{fig1}, we see that setting $\alpha=0.05$
leads to the best computational efficiency, i.e. the R-ADMM
with $\alpha=0.05$ requires the least amount of time to
converge. This is because, for a smaller value of $\alpha$,
although R-ADMM requires more iterations to converge, its
per-iteration complexity is also lower since fewer
subproblems need to be solved at each iteration.

Fig. \ref{fig2} plots the convergence behavior of the proposed
algorithm under different maximum transmit powers, where we
set $\alpha=0.05$. In this experiment, the number of antennas
at each AP is fixed as $36$ and the number of users is set to
$100$. We see that, as expected, increasing the transmit
power results in a higher achievable rate. Fig. \ref{fig3}$(a)$
plots the convergence behavior of the proposed algorithm
when the number of users $K$ varies. In this experiment,
the number of antennas at each AP is fixed as $36$ and the
maximum per-AP transmit power $p_m$ is set to $10$mW. The
sampling rate is fixed as $\alpha=0.05$. We can see that a
smaller number of users leads to a higher achievable rate.
It is also observed that the required runtime scales nearly
linearly with the number of users $K$. Fig. \ref{fig3}$(b)$
plots the convergence behavior of the proposed algorithm
when the number of antennas at the AP varies.

\begin{figure}[!htbp]
    \centering
    \includegraphics[height=4.5cm,width=6cm]{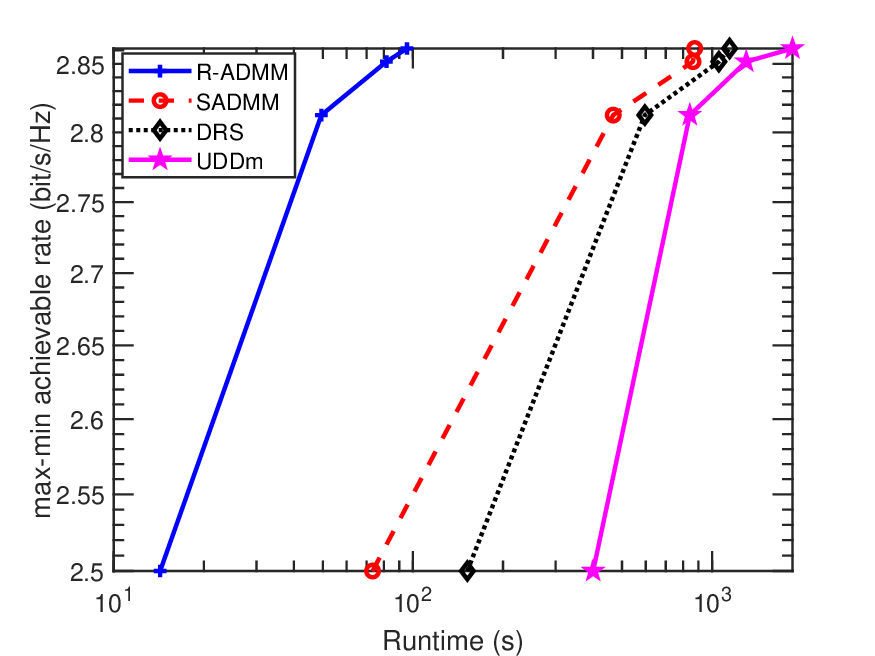}
    \color{blue}\caption{ Computational efficiency of respective algorithms}
    \label{fig4}
\end{figure}

Lastly, we compare our proposed algorithm with
the SADMM, the DRS \cite{WangWang21} as well as the UDDm
(uplink-downlink duality-based method) \cite{YuLan07}.
Note that the UDDm proposed in \cite{YuLan07} is introduced
in Section \ref{sec-relation-udd}. The resulting SDP
(\ref{ref1-uplink}) is solved via the interior-point
method \cite{YuLan07}.
Fig. \ref{fig4} plots the convergence behavior of
respective algorithms, where the number of antennas at
each AP is fixed as $36$ and the number of users is set
to $100$, and the maximum per-AP transmit power is set
to $p_m=10$mW, $1\leq m\leq M$. The sampling rate for
the proposed R-ADMM is set to $\alpha=0.05$. The
parameters of competing algorithms are tuned to achieve
the best performance. From Fig. \ref{fig4}, we can see
that all algorithms converge to the same value. This
means that they have correctly identified the feasibility
or infeasibility of the problem (\ref{feasibility-problem}).
We see that the runtime required by the proposed R-ADMM
to converge is far less than the time required by the
SADMM, the DRS as well as the UDDm. Our simulation
results suggest that the number of iterations required
by R-ADMM with $\alpha=0.05$ is only two or three times
that of SADMM. Nevertheless, since the per-iteration
computational complexity of R-ADMM is approximately
$1/10$ of that of the SADMM, it thus takes R-ADMM
much less time to converge as compared with the
SADMM. We see that DRS performs worse than SADMM. This
is because the performance of DRS is critically dependent
on the geometry of the feasibility-checking problem. If
the intersection of two sets has a large intersection
area, then the DRS can quickly converge to a point
in the intersection area. On the other hand, if the
intersection area is small, the efficiency of DRS incurs
a substantial amount of degradation. At last, it is also
observed that the UDDm takes a larger amount of time to
converge than other competing algorithms. This is
mainly attributed to the fact that the interior-point
method incurs an excessively high computational
complexity, especially for large-scale problems.

\section{Conclusions}
\label{sec-conclusion} In this paper, we proposed a randomized
ADMM algorithm to deal with large-scale max-min beamforming
problems for CF massive MIMO systems. The proposed method is
based on a novel formulation that converts the feasibility-checking
problem into a linearly constrained optimization problem.
Using the orthogonality inherent in the linear constraint,
the most computationally intensive subproblem in our ADMM
algorithm is decomposed into a number of subtasks. In each
iteration, only a small number of subtasks need to be solved,
thus leading to a much lower per-iteration computational
complexity. The proposed algorithm is proved to possess an
$O(1/\bar{t})$ convergence rate, which is in the same order
as its deterministic counterpart. Numerical results show that
the proposed algorithm is significantly more advantageous than
state-of-the-art methods in terms of computational efficiency.

\appendices
%\section{The Structure of $\boldsymbol{P}$}
%\label{appendix-structure-P}
%The purpose of $\boldsymbol{P}$ is to take $\tilde{\boldsymbol{v}}$
%as input and then output the following permuted vector
%\begin{align}
%\begin{bmatrix}
%(\boldsymbol{v}_1[1])^H &
%\end{bmatrix}
%\end{align}

\section{Proof of Proposition \ref{proposition-2}}
\label{appendix-A-1}
First, notice that the objective function in (\ref{ADMM-6}) is
strongly convex, which implies that its solution
$\boldsymbol{w}_i^{t}$ is unique. If $\boldsymbol{d}_i^{t-1}
\in\mathcal{C}_i$, then setting $\boldsymbol{w}_i=
\boldsymbol{d}_i^{t-1}$ results in the minimum value of
the objective function in (\ref{ADMM-6}). Thus,
$\boldsymbol{w}_i^{t}=\boldsymbol{d}_i^{t-1}$ if
$\boldsymbol{d}_i^{t-1}\in\mathcal{C}_i$. Now, if
$\boldsymbol{d}_i^{t-1}\notin\mathcal{C}_i$, set
$\boldsymbol{w}_i=\frac{\beta\boldsymbol{d}_i^{t-1}}{1+\beta}+
\frac{\text{Proj}_{\mathcal{C}_i}\{\boldsymbol{d}_i^{t-1}\}}{1+\beta}$.
We see that the gradient of the objective function is equal to
$\boldsymbol{0}$, that is,
\begin{align}
&\underbrace{\boldsymbol{w}_i
-\text{Proj}_{\mathcal{C}_i}\{\boldsymbol{w}_i\}}_{=\nabla f_i}
+\beta(\boldsymbol{w}_i-\boldsymbol{d}_i^{t-1})
\nonumber\\
&\textstyle\overset{(a)}{=}
\text{Proj}_{\mathcal{C}_i}\{\boldsymbol{d}_i^{t-1}\}-
\text{Proj}_{\mathcal{C}_i}\{\boldsymbol{w}_i\}
\overset{(b)}{=}\boldsymbol{0},
\end{align}
where $(a)$ is because $\boldsymbol{w}_i^{t}+
\beta(\boldsymbol{w}_i^{t}-\boldsymbol{d}_i^{t-1})=
\text{Proj}_{\mathcal{C}_i}\{\boldsymbol{d}_i^{t-1}\}$, and
$(b)$ is because $\text{Proj}_{\mathcal{C}_i}
\{\boldsymbol{w}_i\}=\text{Proj}_{\mathcal{C}_i}
\{\boldsymbol{d}_i^{t-1}\}$ (recall that
$\boldsymbol{w}_i$ is a point lying in the segment
between $\boldsymbol{d}_i^{t-1}$ and $\text{Proj}_{\mathcal{C}_i}
\{\boldsymbol{d}_i^{t-1}\}$, so its projection onto $\mathcal{C}_i$
is equal to $\text{Proj}_{\mathcal{C}_i}\{\boldsymbol{d}_i^{t}\}$).
Since the solution to (\ref{ADMM-6}) is unique, there is
only one point that enables $\boldsymbol{0}$ gradient of
the objective function. Therefore it holds $\boldsymbol{w}_i^{t}
=\frac{\beta\boldsymbol{d}_i^{t-1}}{1+\beta}+
\frac{\text{Proj}_{\mathcal{C}_i}\{\boldsymbol{d}_i^{t-1}\}}{1+\beta}$
when $\boldsymbol{d}_i^{t-1}\notin\mathcal{C}_i$.

\section{Proof of Theorem \ref{theorem-content-2} }
\label{appendix-B}
Consider the $\boldsymbol{w}^{t}$-subproblem, $t<\bar{t}$. Setting
the gradient of the objective function to $\boldsymbol{0}$ yields
\begin{align}
&\textstyle\boldsymbol{0}=\nabla f(\boldsymbol{w}^{t})+
\alpha\beta(\underbrace{\textstyle\boldsymbol{w}^{t}-
\sum_{j=1}^K \boldsymbol{A}_j\tilde{\boldsymbol{v}}_j^{t}
+\boldsymbol{b}}_{\overset{(\ref{sec2-2-1})}{=}
\frac{1}{\beta}(\boldsymbol{\lambda}^{t-1}-
\bar{\boldsymbol{\lambda}}^{t})\overset{(\ref{sec2-2})}{=}
\frac{1}{\alpha\beta}(\boldsymbol{\lambda}^{t-1}-
\boldsymbol{\lambda}^{t})}+
\frac{1}{\alpha\beta}\boldsymbol{\lambda}^{t-1})
\nonumber\\
&\quad+\bar{\alpha}\beta(\boldsymbol{w}^{t}-\boldsymbol{w}^{t-1})
%\nonumber\\
%&\textstyle\overset{(a)}{=}\nabla f(\boldsymbol{w}^{t})+
%\alpha\beta(\frac{1}{\beta}(\boldsymbol{\lambda}^{t-1}-
%\bar{\boldsymbol{\lambda}}^{t})-
%\frac{1}{\alpha\beta}\boldsymbol{\lambda}^{t-1}))+
%\bar{\alpha}\beta(\boldsymbol{w}^{t}-\boldsymbol{w}^{t-1})
%\nonumber\\
%&\textstyle\overset{(\ref{sec2-2})}{=}\nabla f(\boldsymbol{w}^{t})+
%\alpha\beta(\frac{1}{\alpha\beta}(\boldsymbol{\lambda}^{t-1}
%-\boldsymbol{\lambda}^{t})-
%\frac{1}{\alpha\beta}\boldsymbol{\lambda}^{t-1}))+
%\bar{\alpha}\beta(\boldsymbol{w}^{t}-\boldsymbol{w}^{t-1})
\nonumber\\
&\textstyle=\nabla f(\boldsymbol{w}^{t})-\boldsymbol{\lambda}^{t}
+\bar{\alpha}\beta(\boldsymbol{w}^{t}-\boldsymbol{w}^{t-1}).
\label{convergence-1}
\end{align}
Similarly, for $t=\bar{t}$ we have
\begin{align}
&\textstyle\boldsymbol{0}
%=\nabla f(\boldsymbol{w}^{\bar{t}})+
%\beta(\boldsymbol{w}^{\bar{t}}-
%(\sum_{j=1}^K \boldsymbol{A}_j\tilde{\boldsymbol{v}}_j^{\bar{t}}
%+\boldsymbol{b}+\frac{1}{\beta}\boldsymbol{\lambda}^{\bar{t}-1}))+
%\hat{\alpha}\beta(\boldsymbol{w}^{\bar{t}}-
%\boldsymbol{w}^{\bar{t}-1})
%\nonumber\\
\textstyle=\nabla f(\boldsymbol{w}^{\bar{t}})-
\bar{\boldsymbol{\lambda}}^{\bar{t}}+
\alpha\bar{\alpha}\beta(\boldsymbol{w}^{\bar{t}}-
\boldsymbol{w}^{\bar{t}-1}).
\label{convergence-2}
\end{align}
Multiplying $\boldsymbol{w}^{*}-\boldsymbol{w}^{t}$ (resp.
$\boldsymbol{w}^{*}-\boldsymbol{w}^{\bar{t}}$) to both
sides of (\ref{convergence-1}) (resp. (\ref{convergence-2}))
yields
\begin{align}
&\textstyle\boldsymbol{0}=(\boldsymbol{w}^{*}
-\boldsymbol{w}^{t})^T(\nabla f(\boldsymbol{w}^{t})-
\boldsymbol{\lambda}^{t})
+\bar{\alpha}\beta\cdot\Xi(\boldsymbol{w}^{t}), \ t<\bar{t},
\nonumber\\
&\textstyle\boldsymbol{0}=(\boldsymbol{w}^{*}
-\boldsymbol{w}^{\bar{t}})^T(\nabla f(\boldsymbol{w}^{\bar{t}})-
\bar{\boldsymbol{\lambda}}^{\bar{t}})
+\alpha\bar{\alpha}\beta\cdot\Xi(\boldsymbol{w}^{\bar{t}}),
\label{convergence-3}
\end{align}
where $\Xi(\boldsymbol{w}^{t})\triangleq(\boldsymbol{w}^{*}
-\boldsymbol{w}^{t})^T(\boldsymbol{w}^{t}-\boldsymbol{w}^{t-1})$,
$1\leq t\leq \bar{t}$. Summing up the above equalities yields
\begin{align}
&\textstyle\boldsymbol{0}=(\boldsymbol{w}^{*}
-\boldsymbol{w}^{\bar{t}})^T(\nabla f(\boldsymbol{w}^{\bar{t}})-
\bar{\boldsymbol{\lambda}}^{\bar{t}})+\sum_{t=1}^{\bar{t}}
\alpha\bar{\alpha}\beta\cdot\Xi(\boldsymbol{w}^{t})
\nonumber\\
&\textstyle+\alpha\sum_{t=1}^{\bar{t}-1}
(\boldsymbol{w}^{*}-\boldsymbol{w}^{t})^T
(\nabla f(\boldsymbol{w}^{t})-\boldsymbol{\lambda}^{t}).
\label{convergence-4}
\end{align}
Then consider the $\tilde{\boldsymbol{v}}^{t}$-subproblem,
$t<\bar{t}$. For $i\in \Lambda^{t}$ we have
\begin{align}
\textstyle \nabla g_i(\tilde{\boldsymbol{v}}_i^{t})\triangleq
\boldsymbol{A}_i^T(\boldsymbol{A}_i\tilde{\boldsymbol{v}}_i^{t}
+\boldsymbol{D}_i(\boldsymbol{b}-\boldsymbol{w}^{t-1}+
\frac{1}{\beta}\boldsymbol{\lambda}^{t-1}))=\boldsymbol{0},
%\label{convergence-5}
\nonumber
\end{align}
which can be compactly written as
\begin{align}
&\textstyle \boldsymbol{r}^{t}\odot\nabla
g(\tilde{\boldsymbol{v}}^{t})
\triangleq
\nonumber\\
&\textstyle
[\big(r_1^{t}\cdot\nabla g_1(\tilde{\boldsymbol{v}}_1^{t})\big)^T \
\cdots \
\big(r_{K}^{t}\cdot\nabla g_{K}(\tilde{\boldsymbol{v}}_{K}^{t})\big)^T]
=\boldsymbol{0},
\label{convergence-6}
\end{align}
where $\boldsymbol{r}^{t}$ is a $0/1$ random binary vector
of length $K$. The $i$th element of $\boldsymbol{r}^{t}$,
namely, $r_i^{t}$, equals to $1$ if $i\in \Lambda^{t}$
($g_i$ is selected). As such, $r_i^{t}$ has a probability
of $\alpha$ (resp. $1-\alpha$) to be $1$ (resp. $0$). Define
$\mathbb{E}_{\boldsymbol{r}^{t}}^{\{\boldsymbol{r}^{t-1}\}}[\cdot]$
as the expectation taken w.r.t. $\boldsymbol{r}^{t}$ while
conditioned on $\{\boldsymbol{r}^{t-1}\}$, where
$\{\boldsymbol{r}^{t-1}\}$ is short for
$\{\boldsymbol{r}^{t'}\}_{t'=1}^{t-1}$. Multiplying
$\boldsymbol{r}^{t}\odot(\tilde{\boldsymbol{v}}^*-
\tilde{\boldsymbol{v}}^{t})$ to both sides of
(\ref{convergence-6}) yields
\begin{align}
&\boldsymbol{0}=
\mathbb{E}_{\boldsymbol{r}^{t}}^{\{\boldsymbol{r}^{t-1}\}}
\big[(\boldsymbol{r}^{t}\odot(\tilde{\boldsymbol{v}}^*-
\tilde{\boldsymbol{v}}^{t}))^T
(\boldsymbol{r}^{t}\odot\nabla g(\tilde{\boldsymbol{v}}^{t}))\big]
\nonumber\\
%&=\textstyle\mathbb{E}_{\boldsymbol{r}^{t}}^{\{\boldsymbol{r}^{t-1}\}}
%\big[\sum_{i=1}^{K}r_i^{t}(\tilde{\boldsymbol{v}}_i^*-
%\tilde{\boldsymbol{v}}_i^{t})^T\nabla g_i(\tilde{\boldsymbol{v}}_i^{t})\big]
%\nonumber\\
&=\textstyle\mathbb{E}_{\boldsymbol{r}^{t}}^{\{\boldsymbol{r}^{t-1}\}}
\big[\sum_{i=1}^{K}\frac{r_i^{t}}{\beta}(\tilde{\boldsymbol{v}}_i^*-
\tilde{\boldsymbol{v}}_i^{t})^T\boldsymbol{A}_i^T\boldsymbol{\lambda}^{t-1}
\nonumber\\
&\qquad\qquad\quad+\underbrace{r_i^{t}
(\tilde{\boldsymbol{v}}_i^*-\tilde{\boldsymbol{v}}_i^{t})^T
\boldsymbol{A}_i^T(\boldsymbol{A}_i\tilde{\boldsymbol{v}}_i^{t}+
\boldsymbol{b}-
\boldsymbol{w}^{t-1})}_{\triangleq [(\ref{convergence-7})-4]}\big]
\nonumber\\
&\overset{(a)}{=}\textstyle\sum_{i=1}^{K}\frac{\alpha}{\beta}
(\tilde{\boldsymbol{v}}_i^*-
\tilde{\boldsymbol{v}}_i^{t-1})^T\boldsymbol{A}_i^T\boldsymbol{\lambda}^{t-1}
\nonumber\\
&\textstyle+\mathbb{E}_{\boldsymbol{r}^{t}}^{\{\boldsymbol{r}^{t-1}\}}
[\sum_{i=1}^{K}\frac{r_i^{t}}{\beta}(\tilde{\boldsymbol{v}}^{t-1}-
\tilde{\boldsymbol{v}}_i^{t})^T\boldsymbol{A}_i^T\boldsymbol{\lambda}^{t-1}
+[(\ref{convergence-7})-4]]
\nonumber\\
&\overset{(b)}{=}\textstyle
\underbrace{\textstyle\sum_{i=1}^{K}\frac{\alpha-1}{\beta}
(\tilde{\boldsymbol{v}}_i^*-\tilde{\boldsymbol{v}}_i^{t-1})^T
\boldsymbol{A}_i^T\boldsymbol{\lambda}^{t-1}}_{\triangleq [(\ref{convergence-7})-1]}
\nonumber\\
&\textstyle
+\underbrace{\textstyle
\mathbb{E}_{\boldsymbol{r}^{t}}^{\{\boldsymbol{r}^{t-1}\}}
\big[\sum_{i=1}^{K}\frac{1}{\beta}(\tilde{\boldsymbol{v}}_i^{*}-
\tilde{\boldsymbol{v}}_i^{t})^T\boldsymbol{A}_i^T
\boldsymbol{\lambda}^{t-1}\big]}_{\triangleq [(\ref{convergence-7})-2]}
+\mathbb{E}_{\boldsymbol{r}^{t}}^{\{\boldsymbol{r}^{t-1}\}}
\{[(\ref{convergence-7})-4]\}
\nonumber\\
%&\textstyle+\mathbb{E}_{\boldsymbol{r}^{t}}^{\{\boldsymbol{r}^{t-1}\}}
%\big[\sum_{i=1}^{K}r_i^{t}(\tilde{\boldsymbol{v}}_i^*-\tilde{\boldsymbol{v}}_i^{t})^T
%\boldsymbol{A}_i^T(\boldsymbol{A}_i\tilde{\boldsymbol{v}}_i^{t}+
%\boldsymbol{b}-\boldsymbol{w}^{t-1})\big]
%\nonumber\\
&=\textstyle[(\ref{convergence-7})-1]+[(\ref{convergence-7})-2]
\nonumber\\
&+\textstyle\mathbb{E}_{\boldsymbol{r}^{t}}^{\{\boldsymbol{r}^{t-1}\}}
\Big[\sum_{i=1}^{K}r_i^{t}(\tilde{\boldsymbol{v}}_i^*
-\tilde{\boldsymbol{v}}_i^{t-1})^T\underbrace{\boldsymbol{A}_i^T
(\boldsymbol{A}_i\tilde{\boldsymbol{v}}_i^{t-1}+\boldsymbol{b}-
\boldsymbol{w}^{t-1})}_{\triangleq \boldsymbol{\chi}_i^{t-1}}
\nonumber\\
&+r_i^{t}(\tilde{\boldsymbol{v}}_i^{t-1}-\tilde{\boldsymbol{v}}_i^{t})^T
\boldsymbol{\chi}_i^{t-1}+r_i^{t}(\tilde{\boldsymbol{v}}_i^*-
\tilde{\boldsymbol{v}}_i^{t})^T
\boldsymbol{A}_i^T\boldsymbol{A}_i(\tilde{\boldsymbol{v}}_i^{t}-
\tilde{\boldsymbol{v}}_i^{t-1})\Big]
\nonumber\\
&\overset{(c)}{=}\textstyle[(\ref{convergence-7})-1]+
[(\ref{convergence-7})-2]+(\alpha-1)\sum\limits_{i=1}^{K}
(\tilde{\boldsymbol{v}}_i^*-\tilde{\boldsymbol{v}}_i^{t-1})^T
\boldsymbol{\chi}_i^{t-1}+
\nonumber\\
&\textstyle\underbrace{\textstyle
\mathbb{E}_{\boldsymbol{r}^{t}}^{\{\boldsymbol{r}^{t-1}\}}
\big[\sum\limits_{i=1}^K(\tilde{\boldsymbol{v}}_i^*-
\tilde{\boldsymbol{v}}_i^{t})^T\boldsymbol{\chi}_i^{t-1}+
r_i^{t}(\tilde{\boldsymbol{v}}_i^*-\tilde{\boldsymbol{v}}_i^{t})^T
\boldsymbol{A}_i^T\boldsymbol{A}_i(\tilde{\boldsymbol{v}}_i^{t}-
\tilde{\boldsymbol{v}}_i^{t-1})\big]}_{\triangleq [(\ref{convergence-7})-3]}
\nonumber\\
&\overset{(d)}{=}\textstyle[(\ref{convergence-7})-1]
+(\alpha-1)\sum_{i=1}^{K}(\tilde{\boldsymbol{v}}_i^*-
\tilde{\boldsymbol{v}}_i^{t-1})^T
\boldsymbol{\chi}_i^{t-1}+
\nonumber\\
&\quad\textstyle
\mathbb{E}_{\boldsymbol{r}^{t}}^{\{\boldsymbol{r}^{t-1}\}}
\big[\sum_{i=1}^{K}\frac{1}{\beta}(\tilde{\boldsymbol{v}}_i^{*}-
\tilde{\boldsymbol{v}}_i^{t})^T\boldsymbol{A}_i^T
\Big(\boldsymbol{\lambda}^{t}+(1-\alpha)\beta
(\sum\limits_{j=1}^K \boldsymbol{A}_{j}\tilde{\boldsymbol{v}}_{j}^{t}
\nonumber\\
&\textstyle
+\boldsymbol{b}-\boldsymbol{w}^{t})\Big)
+(\tilde{\boldsymbol{v}}_i^{*}-\tilde{\boldsymbol{v}}_i^{t})^T
\boldsymbol{A}_i^T(\boldsymbol{w}^{t}-\boldsymbol{w}^{t-1})\big],
\label{convergence-7}
\end{align}
where $(a)$ is because
$\mathbb{E}_{\boldsymbol{r}^{t}}^{\{\boldsymbol{r}^{t-1}\}}[r_i^{t}]
=\alpha$, and $(b)$ is because (recall that $r_i^{t}$ is $0/1$
binary and $\tilde{\boldsymbol{v}}_i^{t}=
\tilde{\boldsymbol{v}}_i^{t-1}$ when $r_i^{t}=0$)
\begin{align}
&\textstyle\mathbb{E}_{\boldsymbol{r}^{t}}^{\{\boldsymbol{r}^{t-1}\}}
[\sum_{i=1}^{K}\frac{r_i^{t}}{\beta}(\tilde{\boldsymbol{v}}_i^{t-1}-
\tilde{\boldsymbol{v}}_i^{t})^T\boldsymbol{A}_i^T
\boldsymbol{\lambda}^{t-1}]
\nonumber\\
&\textstyle
=\mathbb{E}_{\boldsymbol{r}^{t}}^{\{\boldsymbol{r}^{t-1}\}}
[\sum_{i=1}^{K}\frac{1}{\beta}(\tilde{\boldsymbol{v}}_i^{t-1}-
\tilde{\boldsymbol{v}}_i^{t})^T\boldsymbol{A}_i^T
\boldsymbol{\lambda}^{t-1}],
\label{convergence-8}
\end{align}
$(c)$ is derived with a similar logic as $(b)$, and $(d)$ is
because
\begin{align}
&[(\ref{convergence-7})-2]+[(\ref{convergence-7})-3]
\nonumber\\
&=\textstyle\mathbb{E}_{\boldsymbol{r}^{t}}^{\{\boldsymbol{r}^{t-1}\}}
\big[\sum_{i=1}^{K}\frac{1}{\beta}(\tilde{\boldsymbol{v}}_i^{*}-
\tilde{\boldsymbol{v}}_i^{t})^T\boldsymbol{A}_i^T(\bar{\boldsymbol{\lambda}}^{t}
-\bar{\boldsymbol{\lambda}}^{t}+\boldsymbol{\lambda}^{t-1})\big]
\nonumber\\
&\quad+[(\ref{convergence-7})-3]
\nonumber\\
&\textstyle\overset{(\ref{sec2-2-1})}{=}\mathbb{E}_{\boldsymbol{r}^{t}}^{\{\boldsymbol{r}^{t-1}\}}
\big[\sum_{i=1}^{K}\frac{1}{\beta}(\tilde{\boldsymbol{v}}_i^{*}-
\tilde{\boldsymbol{v}}_i^{t})^T(\boldsymbol{A}_i^T\bar{\boldsymbol{\lambda}}^{t}
\nonumber\\
&\qquad\qquad\quad\textstyle-\beta\boldsymbol{A}_i^T(\sum_{j=1}^K \boldsymbol{A}_{j}\tilde{\boldsymbol{v}}_{j}^{t}
+\boldsymbol{b}-\boldsymbol{w}^{t}))\big]+[(\ref{convergence-7})-3]
\nonumber\\
&\textstyle\overset{(e)}{=}\mathbb{E}_{\boldsymbol{r}^{t}}^{\{\boldsymbol{r}^{t-1}\}}
\big[\sum_{i=1}^{K}(\tilde{\boldsymbol{v}}_i^{*}-
\tilde{\boldsymbol{v}}_i^{t})^T(\frac{1}{\beta}\boldsymbol{A}_i^T\bar{\boldsymbol{\lambda}}^{t}
-\boldsymbol{A}_i^T(\boldsymbol{A}_i\tilde{\boldsymbol{v}}_i^{t}
+\boldsymbol{b}-
\nonumber\\
&\quad\boldsymbol{w}^{t}))\big]+[(\ref{convergence-7})-3]
\nonumber\\
&\textstyle=\mathbb{E}_{\boldsymbol{r}^{t}}^{\{\boldsymbol{r}^{t-1}\}}
\Big[\sum_{i=1}^{K}(\tilde{\boldsymbol{v}}_i^{*}-
\tilde{\boldsymbol{v}}_i^{t})^T\Big(\frac{1}{\beta}\boldsymbol{A}_i^T\bar{\boldsymbol{\lambda}}^{t}
\nonumber\\
&\qquad\quad\textstyle-(1-r_i^{t})\boldsymbol{A}_i^T\boldsymbol{A}_i(\tilde{\boldsymbol{v}}_i^{t}-
\tilde{\boldsymbol{v}}_i^{t-1})+\boldsymbol{A}_i^T(\boldsymbol{w}^{t}-\boldsymbol{w}^{t-1})\Big)\Big]
\nonumber\\
%&\quad+\textstyle\mathbb{E}_{\boldsymbol{r}^{t}}^{\{\boldsymbol{r}^{t-1}\}}
%\big[\sum_{i=1}^{K}r_i^{t}(\tilde{\boldsymbol{v}}_i^*-\tilde{\boldsymbol{v}}_i^{t})^T
%\boldsymbol{A}_i^T\boldsymbol{A}_i(\tilde{\boldsymbol{v}}_i^{t}-
%\tilde{\boldsymbol{v}}_i^{t-1})\big]
%\nonumber\\
%&\textstyle\overset{(f)}{=}
%\mathbb{E}_{\boldsymbol{r}^{t}}^{\{\boldsymbol{r}^{t-1}\}}
%\big[\sum_{i=1}^{K}(\tilde{\boldsymbol{v}}_i^{*}-
%\tilde{\boldsymbol{v}}_i^{t})^T(\frac{1}{\beta}\boldsymbol{A}_i^T\bar{\boldsymbol{\lambda}}^{t}
%+\boldsymbol{A}_i^T(\boldsymbol{w}^{t}-\boldsymbol{w}^{t-1}))\big]
%\nonumber\\
&\textstyle\overset{(f)}{=}\mathbb{E}_{\boldsymbol{r}^{t}}^{\{\boldsymbol{r}^{t-1}\}}
\Big[\sum_{i=1}^{K}(\tilde{\boldsymbol{v}}_i^{*}-
\tilde{\boldsymbol{v}}_i^{t})^T\boldsymbol{A}_i^T\Big(\frac{1}{\beta}\Big(\boldsymbol{\lambda}^{t}+
\nonumber\\
&\quad\textstyle(1-\alpha)\beta(\sum\limits_{j=1}^K \boldsymbol{A}_{j}\tilde{\boldsymbol{v}}_{j}^{t}
+\boldsymbol{b}-\boldsymbol{w}^{t})\Big)+(\boldsymbol{w}^{t}-\boldsymbol{w}^{t-1})\Big)\Big],
%\nonumber\\
%&\qquad\quad+(\tilde{\boldsymbol{v}}_i^{*}-\tilde{\boldsymbol{v}}_i^{t})^T
%\boldsymbol{A}_i^T(\boldsymbol{w}^{t}-\boldsymbol{w}^{t-1})\big],
\label{convergence-9}
\end{align}
in which $(e)$ is because $\boldsymbol{A}_i^T\boldsymbol{A}_{j}=
\boldsymbol{0}$, $i\neq j$, $(f)$ is obtained
by noticing $\mathbb{E}_{\boldsymbol{r}^{t}}^{\{\boldsymbol{r}^{t-1}\}}
\big[\sum_{i=1}^{K}(1-r_i^{t})(\tilde{\boldsymbol{v}}_i^*-
\tilde{\boldsymbol{v}}_i^{t})^T\boldsymbol{A}_i^T
\boldsymbol{A}_i(\tilde{\boldsymbol{v}}_i^{t}-
\tilde{\boldsymbol{v}}_i^{t-1})\big]=0$ (recall that $r_i^{t}$ is $0/1$
binary and $\tilde{\boldsymbol{v}}_i^{t}=\tilde{\boldsymbol{v}}_i^{t-1}$
when $r_i^{t}=0$) as well as
%\begin{align}
%\textstyle\mathbb{E}_{\boldsymbol{r}^{t}}^{\{\boldsymbol{r}^{t-1}\}}
%\big[\sum\limits_{i=1}^{K}r_i^{t}(\tilde{\boldsymbol{v}}_i^*-\tilde{\boldsymbol{v}}_i^{t})^T
%\boldsymbol{A}_i^T\boldsymbol{A}_i(\tilde{\boldsymbol{v}}_i^{t}-
%\tilde{\boldsymbol{v}}_i^{t-1})\big]=
%\mathbb{E}_{\boldsymbol{r}^{t}}^{\{\boldsymbol{r}^{t-1}\}}
%\big[\sum\limits_{i=1}^{K}(\tilde{\boldsymbol{v}}_i^*-\tilde{\boldsymbol{v}}_i^{t})^T
%\boldsymbol{A}_i^T\boldsymbol{A}_i(\tilde{\boldsymbol{v}}_i^{t}-
%\tilde{\boldsymbol{v}}_i^{t-1})\big],
%\label{convergence-10}
%\end{align}
\begin{align}
&\textstyle\bar{\boldsymbol{\lambda}}^{t}\overset{(\ref{sec2-2})}{=}
\boldsymbol{\lambda}^{t}+(1-\alpha^{-1})(\boldsymbol{\lambda}^{t-1}
-\boldsymbol{\lambda}^{t})
\nonumber\\
&\textstyle\overset{(\ref{sec2-2})}{=}\boldsymbol{\lambda}^{t}-
(1-\alpha^{-1})(\alpha\beta(\sum_{j=1}^K \boldsymbol{A}_j\tilde{\boldsymbol{v}}_j^{t}
+\boldsymbol{b}-\boldsymbol{w}^{t})).
\label{convergence-10}
\end{align}
Define
\begin{align}
&\textstyle M^{t}\triangleq\sum\limits_{i=1}^{K}(\tilde{\boldsymbol{v}}_i^*-
\tilde{\boldsymbol{v}}_i^{t})^T\boldsymbol{A}_i^T\boldsymbol{\lambda}^{t}
=(\tilde{\boldsymbol{v}}^*-\tilde{\boldsymbol{v}}^{t})^T
\boldsymbol{A}^T\boldsymbol{\lambda}^{t}, \ 1\leq t\leq \bar{t},
\nonumber\\
&\textstyle V^{t}\triangleq\beta\sum_{i=1}^{K}(\tilde{\boldsymbol{v}}_i^*-
\tilde{\boldsymbol{v}}_i^{t})^T
\boldsymbol{A}_i^T(\sum_{j=1}^K\boldsymbol{A}_{j}\tilde{\boldsymbol{v}}_{j}^{t}+
\boldsymbol{b}-\boldsymbol{w}^{t})
\nonumber\\
&\quad\textstyle=\beta(\tilde{\boldsymbol{v}}^*-\tilde{\boldsymbol{v}}^{t})^T
\boldsymbol{A}^T(\sum_{j=1}^K\boldsymbol{A}_{j}\tilde{\boldsymbol{v}}_{j}^{t}+
\boldsymbol{b}-\boldsymbol{w}^{t}), \ 1\leq t\leq \bar{t},
\nonumber\\
&\textstyle Y^{t}\triangleq\sum_{i=1}^K\beta(\tilde{\boldsymbol{v}}_i^{*}-
\tilde{\boldsymbol{v}}_i^{t})^T
\boldsymbol{A}_i^T(\boldsymbol{w}^{t}-\boldsymbol{w}^{t-1})
\nonumber\\
&\textstyle
\quad=\beta(\tilde{\boldsymbol{v}}^{*}-\tilde{\boldsymbol{v}}^{t})^T
\boldsymbol{A}^T(\boldsymbol{w}^{t}-\boldsymbol{w}^{t-1}), \ 1\leq t\leq \bar{t}.
\label{convergence-11}
\end{align}
Then (\ref{convergence-7}) can be compactly written as
\begin{align}
&\boldsymbol{0}=(\alpha-1)(\textstyle M^{t-1}+V^{t-1})
+\mathbb{E}_{\boldsymbol{r}^{t}}^{\{\boldsymbol{r}^{t-1}\}}
\big[M^{t}+(1-\alpha)V^{t}+Y^{t}\big].
\nonumber
%\label{convergence-12}
\end{align}
Taking a full expectation for the above yields
\begin{align}
&\boldsymbol{0}=\mathbb{E}_{\{\boldsymbol{r}^{t}\}}
\big[(\alpha-1)(\textstyle M^{t-1}+V^{t-1})+M^{t}+(1-\alpha)V^{t}
+Y^{t}\big].
\nonumber
%\label{convergence-12}
\end{align}
Summing up the above equation for all $t$, we obtain
\begin{align}
&\textstyle\boldsymbol{0}=(\alpha-1)(\textstyle M^{0}+V^0)
\nonumber\\
&\quad\textstyle+\mathbb{E}_{\{\boldsymbol{r}^{\bar{t}}\}}
\big[M^{\bar{t}}+(1-\alpha)V^{\bar{t}}+\alpha\sum_{t=1}^{\bar{t}-1}M^{t}
+\sum_{t=1}^{\bar{t}}Y^{t}\big]
\nonumber\\
&\textstyle\overset{(a)}{=}
\text{Const}
+\mathbb{E}_{\{\boldsymbol{r}^{\bar{t}}\}}
\big[(\tilde{\boldsymbol{v}}^*-\tilde{\boldsymbol{v}}^{\bar{t}})^T
\boldsymbol{A}^T\bar{\boldsymbol{\lambda}}^{\bar{t}}+
\alpha\sum\limits_{t=1}^{\bar{t}-1}M^{t}+\sum\limits_{t=1}^{\bar{t}}Y^{t}\big],
\label{convergence-13}
\end{align}
where $\text{Const}\triangleq(\alpha-1)(\textstyle M^{0}+V^0)$,
$(a)$ is because
\begin{align}
&\textstyle M^{\bar{t}}+(1-\alpha)V^{\bar{t}}
%\overset{(\ref{convergence-11})}{=}
%(\tilde{\boldsymbol{v}}^*-\tilde{\boldsymbol{v}}^{\bar{t}})^T
%\boldsymbol{A}^T\big(\boldsymbol{\lambda}^{\bar{t}}+
%(1-\alpha)\beta(\sum_{j=1}^K\boldsymbol{A}_{j}\tilde{\boldsymbol{v}}_{j}^{\bar{t}}+
%\boldsymbol{b}-\boldsymbol{w}^{\bar{t}})\big)
%\nonumber\\
\overset{(\ref{sec2-2-1})}{=}
\textstyle(\tilde{\boldsymbol{v}}^*-
\tilde{\boldsymbol{v}}^{\bar{t}})^T
\boldsymbol{A}^T\big(\boldsymbol{\lambda}^{\bar{t}}+
(1-\alpha)(\bar{\boldsymbol{\lambda}}^{\bar{t}}-
\boldsymbol{\lambda}^{\bar{t}-1})\big)
\nonumber\\
&\overset{(a)}{=}\textstyle(\tilde{\boldsymbol{v}}^*-
\tilde{\boldsymbol{v}}^{\bar{t}})^T
\boldsymbol{A}^T\bar{\boldsymbol{\lambda}}^{\bar{t}},
%\nonumber\\
%&\overset{(\ref{sec2-2})}{=}\textstyle(\tilde{\boldsymbol{v}}^*-
%\tilde{\boldsymbol{v}}^{\bar{t}})^T
%\boldsymbol{A}^T\big(\boldsymbol{\lambda}^{\bar{t}-1}+
%\alpha(\bar{\boldsymbol{\lambda}}^{\bar{t}}
%-\boldsymbol{\lambda}^{\bar{t}-1})+(1-\alpha)
%(\bar{\boldsymbol{\lambda}}^{\bar{t}}-\boldsymbol{\lambda}^{\bar{t}-1})\big)
\label{convergence-14}
\end{align}
where $(a)$ is because $\boldsymbol{\lambda}^{\bar{t}}
=\boldsymbol{\lambda}^{\bar{t}-1}+\alpha(\bar{\boldsymbol{\lambda}}^{\bar{t}}
-\boldsymbol{\lambda}^{\bar{t}-1})$. Using (\ref{sec2-2-1})
and (\ref{sec2-2}), we can also deduce that
\begin{align}
&\textstyle\boldsymbol{0}=(\bar{\boldsymbol{\lambda}}^{\bar{t}}-
\boldsymbol{\lambda})^T
\textstyle\big(\beta^{-1}(\boldsymbol{\lambda}^{\bar{t}-1}-
\bar{\boldsymbol{\lambda}}^{\bar{t}})+
(\sum_{j=1}^K \boldsymbol{A}_j\tilde{\boldsymbol{v}}_j^{\bar{t}}
+\boldsymbol{b}-\boldsymbol{w}^{\bar{t}})\big),
\nonumber\\
&\textstyle\boldsymbol{0}=\alpha(\boldsymbol{\lambda}^{t}-
\boldsymbol{\lambda})^T\textstyle((\alpha\beta)^{-1}
(\boldsymbol{\lambda}^{t-1}-\boldsymbol{\lambda}^{t})+
(\sum\limits_{j=1}^K \boldsymbol{A}_j\tilde{\boldsymbol{v}}_j^{t}
+\boldsymbol{b}-\boldsymbol{w}^{t})),
\nonumber\\
&\textstyle \qquad 1\leq t\leq \bar{t}-1.
\label{convergence-15}
\end{align}
Adding (\ref{convergence-15}) and (\ref{convergence-4}) to
(\ref{convergence-13}) yields
\begin{align}
%\underbrace{(\alpha-1)(\textstyle M^{0}+V^0)}_{\triangleq\text{Const}}
%+\mathbb{E}_{\{\boldsymbol{r}^{\bar{t}}\}}
%\big[(\tilde{\boldsymbol{v}}^*-\tilde{\boldsymbol{v}}^{\bar{t}})^T
%\boldsymbol{A}^T\bar{\boldsymbol{\lambda}}^{\bar{t}}+
%\alpha\sum_{t=1}^{\bar{t}-1}M^{t}+\sum_{t=1}^{\bar{t}}Y^{t}\big]
%\nonumber\\
%&\quad\textstyle+(\bar{\boldsymbol{\lambda}}^{\bar{t}}-
%\boldsymbol{\lambda})^T\big(\beta^{-1}(\boldsymbol{\lambda}^{\bar{t}-1}-
%\bar{\boldsymbol{\lambda}}^{\bar{t}})+
%(\sum_{j=1}^K \boldsymbol{A}_j\tilde{\boldsymbol{v}}_j^{\bar{t}}
%+\boldsymbol{b}-\boldsymbol{w}^{\bar{t}})\big)
%\nonumber\\
%&\quad\textstyle+\sum_{t=1}^{\bar{t}-1}\alpha(\boldsymbol{\lambda}^{t}
%-\boldsymbol{\lambda})^T((\alpha\beta)^{-1}(\boldsymbol{\lambda}^{t-1}
%-\boldsymbol{\lambda}^{t})+
%(\sum_{j=1}^K \boldsymbol{A}_j\tilde{\boldsymbol{v}}_j^{t}
%+\boldsymbol{b}-\boldsymbol{w}^{t}))
%\nonumber\\
%&\quad\textstyle+(\boldsymbol{w}^{*}
%-\boldsymbol{w}^{\bar{t}})^T(\nabla f(\boldsymbol{w}^{\bar{t}})-
%\boldsymbol{\lambda}^{\bar{t}})+\alpha\sum_{t=1}^{\bar{t}-1}(\boldsymbol{w}^{*}
%-\boldsymbol{w}^{t})^T(\nabla f(\boldsymbol{w}^{t})-\boldsymbol{\lambda}^{t})
%+\sum_{t=1}^{\bar{t}}\alpha\bar{\alpha}\beta\cdot\Xi(\boldsymbol{w}^{t})
%\nonumber\\
&0=\textstyle \text{Const}+
\underbrace{\textstyle\mathbb{E}_{\{\boldsymbol{r}^{\bar{t}}\}}
\Big[\frac{1}{\beta}(\bar{\boldsymbol{\lambda}}^{\bar{t}}
-\boldsymbol{\lambda})^T(\boldsymbol{\lambda}^{\bar{t}-1}
-\bar{\boldsymbol{\lambda}}^{\bar{t}})+}_{[(\ref{convergence-16})-1]}
\nonumber\\
&\underbrace{\textstyle\sum\limits_{t=1}^{\bar{t}-1}\frac{1}{\beta}
(\boldsymbol{\lambda}^{t}
-\boldsymbol{\lambda})^T(\boldsymbol{\lambda}^{t-1}
-\boldsymbol{\lambda}^{t})+\sum\limits_{t=1}^{\bar{t}}\big(Y^{t}+
\alpha\bar{\alpha}\beta\cdot\Xi(\boldsymbol{w}^{t})\big)
\Big]}_{[(\ref{convergence-16})-1]}+
\nonumber\\
&\textstyle\mathbb{E}_{\{\boldsymbol{r}^{\bar{t}}\}}
\big[\underbrace{\textstyle(\tilde{\boldsymbol{v}}^*-
\tilde{\boldsymbol{v}}^{\bar{t}})^T
\boldsymbol{A}^T\bar{\boldsymbol{\lambda}}^{\bar{t}}+
(\bar{\boldsymbol{\lambda}}^{\bar{t}}-\boldsymbol{\lambda})^T
\big(\sum\limits_{j=1}^K \boldsymbol{A}_j\tilde{\boldsymbol{v}}_j^{\bar{t}}
+\boldsymbol{b}-\boldsymbol{w}^{\bar{t}}\big)}_{[(\ref{convergence-16})-2]}
\nonumber\\
&\underbrace{+\textstyle(\boldsymbol{w}^{*}-\boldsymbol{w}^{\bar{t}})^T
(\nabla f(\boldsymbol{w}^{\bar{t}})-
\bar{\boldsymbol{\lambda}}^{\bar{t}})}_{[(\ref{convergence-16})-2]}
\big]+
\nonumber\\
&\textstyle\mathbb{E}_{\{\boldsymbol{r}^{\bar{t}}\}}
\Big[\underbrace{\textstyle\alpha\sum_{t=1}^{\bar{t}-1}
\Big(M^{t}+(\boldsymbol{\lambda}^{t}-\boldsymbol{\lambda})^T\big(\sum_{j=1}^K
\boldsymbol{A}_j\tilde{\boldsymbol{v}}_j^{t}
+\boldsymbol{b}-\boldsymbol{w}^{t}\big)}_{[(\ref{convergence-16})-3]}
\nonumber\\
&\underbrace{+(\boldsymbol{w}^{*}-\boldsymbol{w}^{t})^T
(\nabla f(\boldsymbol{w}^{t})-
\boldsymbol{\lambda}^{t})\Big)}_{[(\ref{convergence-16})-3]}\Big]
\label{convergence-16}
\end{align}
Regarding the term $[(\ref{convergence-16})-1]$, we have
(see Appendix \ref{appendix-C})
\begin{align}
&\textstyle\mathbb{E}_{\{\boldsymbol{r}^{\bar{t}}\}}
\big[[(\ref{convergence-16})-1]\big]\leq -\frac{1}{2\beta}
\big(\|\bar{\boldsymbol{\lambda}}^{\bar{t}}-\boldsymbol{\lambda}\|_2^2
+\|\boldsymbol{\lambda}^{\bar{t}-1}-\bar{\boldsymbol{\lambda}}^{\bar{t}}\|_2^2
\nonumber\\
&\textstyle-\|\boldsymbol{\lambda}^{\bar{t}-1}-\boldsymbol{\lambda}\|_2^2\big)
-\frac{1}{2\beta}\sum_{t=1}^{\bar{t}-1}
(\|\boldsymbol{\lambda}^{t}-\boldsymbol{\lambda}\|_2^2
+\|\boldsymbol{\lambda}^{t-1}-\boldsymbol{\lambda}^{t}\|_2^2
\nonumber\\
&\textstyle
-\|\boldsymbol{\lambda}^{t-1}-\boldsymbol{\lambda}\|_2^2)
-\frac{\alpha\bar{\alpha}\beta\cdot}{2}
\sum_{t=1}^{\bar{t}}(\|\boldsymbol{w}^{t}-\boldsymbol{w}^*\|_2^2
+\|\boldsymbol{w}^{t-1}-\boldsymbol{w}^{t}\|_2^2
\nonumber\\
&\textstyle
-\|\boldsymbol{w}^{t-1}-\boldsymbol{w}^*\|_2^2)
+\frac{1}{2\beta}\|\boldsymbol{\lambda}^{\bar{t}-1}-
\bar{\boldsymbol{\lambda}}^{\bar{t}}\|_2^2
-\frac{\beta}{2}\big(\|\boldsymbol{w}^{*}-
\boldsymbol{w}^{\bar{t}}\|_2^2-
\nonumber\\
&\textstyle
\|\boldsymbol{w}^{*}-\boldsymbol{w}^{\bar{t}-1}\|_2^2\big)+
\sum_{t=1}^{\bar{t}-1}\Big(\frac{1}{2\beta}
\|\boldsymbol{\lambda}^{t-1}-\boldsymbol{\lambda}^{t}\|_2^2+
(\frac{\beta}{2\alpha^2}-\frac{\beta}{2})\|\boldsymbol{w}^{t}-
\nonumber\\
&\textstyle
\boldsymbol{w}^{t-1}\|_2^2
-\frac{\beta}{2}\big(\|\boldsymbol{w}^{*}-\boldsymbol{w}^{t}\|_2^2-
\|\boldsymbol{w}^{*}-\boldsymbol{w}^{t-1}\|_2^2\big)\Big).
\label{convergence-17}
\end{align}
Reorganizing the terms yields
\begin{align}
&\textstyle\mathbb{E}_{\{\boldsymbol{r}^{\bar{t}}\}}
\big[[(\ref{convergence-16})-1]\big]\leq
\frac{1}{2\beta}\|\boldsymbol{\lambda}^{0}-\boldsymbol{\lambda}\|_2^2
-\frac{\beta}{2}\big(\|\boldsymbol{w}^{\bar{t}}-
\boldsymbol{w}^{\bar{t}-1}\|_2^2-
\nonumber\\
&\textstyle\|\boldsymbol{w}^{*}-\boldsymbol{w}^{\bar{t}-1}\|_2^2\big)
+\sum_{t=1}^{\bar{t}-1}\Big((\frac{\beta}{2\alpha^2}-
\frac{\beta}{2})\|\boldsymbol{w}^{t}-\boldsymbol{w}^{t-1}\|_2^2-
\nonumber\\
&\textstyle\frac{\beta}{2}\big(\|\boldsymbol{w}^{*}-\boldsymbol{w}^{t}\|_2^2-
\|\boldsymbol{w}^{*}-\boldsymbol{w}^{t-1}\|_2^2\big)\Big)
-\frac{\alpha\bar{\alpha}\beta\cdot}{2}
\sum\limits_{t=1}^{\bar{t}}(\|\boldsymbol{w}^{t}-\boldsymbol{w}^*\|_2^2
\nonumber\\
&\textstyle
+\|\boldsymbol{w}^{t-1}-\boldsymbol{w}^{t}\|_2^2
-\|\boldsymbol{w}^{t-1}-\boldsymbol{w}^*\|_2^2)
\nonumber\\
&\textstyle\leq \frac{1}{2\beta}\|\boldsymbol{\lambda}^{0}-
\boldsymbol{\lambda}\|_2^2
+\sum_{t=1}^{\bar{t}-1}\big((\frac{\beta}{2\alpha^2}-
\frac{\beta}{2}-\frac{\alpha\bar{\alpha}\beta}{2})\cdot
\|\boldsymbol{w}^{t}-\boldsymbol{w}^{t-1}\|_2^2\big)
\nonumber\\
&\textstyle-\sum_{t=1}^{\bar{t}-1}\frac{\beta}{2}\big(
\|\boldsymbol{w}^{*}-\boldsymbol{w}^{t}\|_2^2-
\|\boldsymbol{w}^{*}-\boldsymbol{w}^{t-1}\|_2^2\big)
\nonumber\\
&\textstyle
-\frac{\alpha\bar{\alpha}\beta}{2}
\big(\sum_{t=1}^{\bar{t}}\|\boldsymbol{w}^{t}-\boldsymbol{w}^*\|_2^2
-\|\boldsymbol{w}^{t-1}-\boldsymbol{w}^*\|_2^2\big)
\nonumber\\
&\textstyle
-\frac{\beta}{2}\big(\|\boldsymbol{w}^{*}-\boldsymbol{w}^{\bar{t}}\|_2^2-
\|\boldsymbol{w}^{*}-\boldsymbol{w}^{\bar{t}-1}\|_2^2\big)
\nonumber\\
&\textstyle\leq \frac{1}{2\beta}\|\boldsymbol{\lambda}^{0}-
\boldsymbol{\lambda}\|_2^2+(\frac{\beta}{2}+
\frac{\alpha\bar{\alpha}\beta}{2})
\|\boldsymbol{w}^{*}-\boldsymbol{w}^{0}\|_2^2
\nonumber\\
&\textstyle
+\sum_{t=1}^{\bar{t}-1}\big((\frac{\beta}{2\alpha^2}-
\frac{\beta}{2}-\frac{\alpha\bar{\alpha}\beta}{2})\cdot
\|\boldsymbol{w}^{t}-\boldsymbol{w}^{t-1}\|_2^2\big)
\nonumber\\
&\textstyle\overset{(a)}{\leq}
\frac{1}{2\beta}\|\boldsymbol{\lambda}^{0}-\boldsymbol{\lambda}\|_2^2
+(\frac{\beta}{2}+\frac{\alpha\bar{\alpha}\beta}{2})
\|\boldsymbol{w}^{*}-\boldsymbol{w}^{0}\|_2^2,
\label{convergence-18}
\end{align}
where $(a)$ is because $\alpha\bar{\alpha}\overset{(\ref{theorem-1})}
{\geq} (\frac{1}{\alpha^2}-1)\Rightarrow(\frac{\beta}{2\alpha^2}-
\frac{\beta}{2})-\frac{\alpha\bar{\alpha}\beta}{2}\leq 0$.
%\begin{align}
%\textstyle \alpha\bar{\alpha}\geq (\frac{1}{\alpha^2}-1)
%\Rightarrow
%(\frac{\beta}{2\alpha^2}-\frac{\beta}{2})-\frac{\alpha\bar{\alpha}\beta}{2}
%\leq 0
%\label{convergence-19}
%\end{align}
Substituting (\ref{convergence-18}) into (\ref{convergence-16}) yields
\begin{align}
\textstyle 0\leq  \text{Const}+
\frac{1}{2\beta}\|\boldsymbol{\lambda}^{0}-\boldsymbol{\lambda}\|_2^2
+\mathbb{E}_{\{\boldsymbol{r}^{\bar{t}}\}}
\big[[(\ref{convergence-16})-2]+\alpha[(\ref{convergence-16})-3]\big],
\label{convergence-20}
\end{align}
where Const has absorbed the constant term $(\frac{\beta}{2}+
\frac{\alpha\bar{\alpha}\beta}{2})\|\boldsymbol{w}^{*}-
\boldsymbol{w}^{0}\|_2^2$. Regarding $(\boldsymbol{w}^{*}-
\boldsymbol{w}^{t})^T\nabla f(\boldsymbol{w}^{t})$ in
$[(\ref{convergence-16})-2]$ and $[(\ref{convergence-16})-3]$,
using the convexity of $f$ we have
\begin{align}
(\boldsymbol{w}^{*}-\boldsymbol{w}^{t})^T\nabla f(\boldsymbol{w}^{t})
\leq f(\boldsymbol{w}^{*})-f(\boldsymbol{w}^{t}).
\label{convergence-21}
\end{align}
Substituting (\ref{convergence-21}) into (\ref{convergence-20})
yields
\begin{align}
&\textstyle 0\leq  \text{Const}+
\frac{1}{2\beta}\|\boldsymbol{\lambda}^{0}-\boldsymbol{\lambda}\|_2^2+
\mathbb{E}_{\{\boldsymbol{r}^{\bar{t}}\}}
\big[\Upsilon^{\bar{t}}+f(\boldsymbol{w}^{*})-f(\boldsymbol{w}^{\bar{t}})
\big]
\nonumber\\
&\textstyle+\alpha\mathbb{E}_{\{\boldsymbol{r}^{\bar{t}}\}}
[\textstyle\sum_{t=1}^{\bar{t}-1}\Upsilon^{t}
+f(\boldsymbol{w}^{*})-f(\boldsymbol{w}^{t})],
\label{convergence-22}
\end{align}
where
\begin{align}
\Upsilon^{\bar{t}}\triangleq&\textstyle(\tilde{\boldsymbol{v}}^*-
\tilde{\boldsymbol{v}}^{\bar{t}})^T
\boldsymbol{A}^T\bar{\boldsymbol{\lambda}}^{\bar{t}}+
(\bar{\boldsymbol{\lambda}}^{\bar{t}}-\boldsymbol{\lambda})^T
\big(\sum_{j=1}^K \boldsymbol{A}_j\tilde{\boldsymbol{v}}_j^{\bar{t}}
+\boldsymbol{b}-\boldsymbol{w}^{\bar{t}}\big)
\nonumber\\
&-(\boldsymbol{w}^{*}-\boldsymbol{w}^{\bar{t}})^T
\bar{\boldsymbol{\lambda}}^{\bar{t}}
\nonumber\\
\Upsilon^{t}\triangleq&\textstyle M^{t}+(\boldsymbol{\lambda}^{t}-
\boldsymbol{\lambda})^T
\big(\sum_{j=1}^K\boldsymbol{A}_j\tilde{\boldsymbol{v}}_j^{t}
+\boldsymbol{b}-\boldsymbol{w}^{t}\big)
\nonumber\\
&-(\boldsymbol{w}^{*}-\boldsymbol{w}^{t})^T\boldsymbol{\lambda}^{t}.
\nonumber
\end{align}
Regarding $\Upsilon^{\bar{t}}$, eliminating those replicative terms
we have
\begin{align}
&\textstyle \Upsilon^{\bar{t}}
%=(\tilde{\boldsymbol{v}}^*-
%\tilde{\boldsymbol{v}}^{\bar{t}})^T
%\boldsymbol{A}^T\bar{\boldsymbol{\lambda}}^{\bar{t}}+
%(\bar{\boldsymbol{\lambda}}^{\bar{t}}-\boldsymbol{\lambda})^T
%\big(\boldsymbol{A}\tilde{\boldsymbol{v}}^{\bar{t}}
%+\boldsymbol{b}-\boldsymbol{w}^{\bar{t}}\big)-
%(\boldsymbol{w}^{*}-\boldsymbol{w}^{\bar{t}})^T
%\bar{\boldsymbol{\lambda}}^{\bar{t}}
%\nonumber\\
%&=(\tilde{\boldsymbol{v}}^*)^T\boldsymbol{A}^T\bar{\boldsymbol{\lambda}}^{\bar{t}}-
%\boldsymbol{\lambda}^T\boldsymbol{A}\tilde{\boldsymbol{v}}^{\bar{t}}+
%(\bar{\boldsymbol{\lambda}}^{\bar{t}}-\boldsymbol{\lambda})^T
%\boldsymbol{b}+\boldsymbol{\lambda}^T\boldsymbol{w}^{\bar{t}}-
%(\boldsymbol{w}^{*})^T\bar{\boldsymbol{\lambda}}^{\bar{t}}
%\nonumber\\
=(\bar{\boldsymbol{\lambda}}^{\bar{t}})^T
(\boldsymbol{A}\tilde{\boldsymbol{v}}^*
+\boldsymbol{b}-\boldsymbol{w}^{*})-
\boldsymbol{\lambda}^T(\boldsymbol{A}\tilde{\boldsymbol{v}}^{\bar{t}}
+\boldsymbol{b}-\boldsymbol{w}^{\bar{t}})
\nonumber\\
&\textstyle
\overset{\boldsymbol{A}\tilde{\boldsymbol{v}}^*
+\boldsymbol{b}-\boldsymbol{w}^{*}=\boldsymbol{0}}{=}
-\boldsymbol{\lambda}^T(\boldsymbol{A}\tilde{\boldsymbol{v}}^{\bar{t}}
+\boldsymbol{b}-\boldsymbol{w}^{\bar{t}}).
\label{convergence-23}
\end{align}
Analogously, for $\Upsilon^{t}$, $1\leq t\leq \bar{t}-1$, we have
\begin{align}
&\textstyle \Upsilon^{t}=-\boldsymbol{\lambda}^T
(\boldsymbol{A}\tilde{\boldsymbol{v}}^{t}
+\boldsymbol{b}-\boldsymbol{w}^{t}).
\label{convergence-24}
\end{align}
Substituting (\ref{convergence-23}) and (\ref{convergence-24})
into (\ref{convergence-22})
\begin{align}
&\textstyle 0\leq  \underbrace{\textstyle\text{Const}+
\frac{1}{2\beta}\|\boldsymbol{\lambda}^{0}-
\boldsymbol{\lambda}\|_2^2}_{[(\ref{convergence-26})-1]}
+\mathbb{E}_{\{\boldsymbol{r}^{\bar{t}}\}}
\big[f(\boldsymbol{w}^{*})-f(\boldsymbol{w}^{\bar{t}})-
\nonumber\\
&\textstyle
\quad\boldsymbol{\lambda}^T(\boldsymbol{A}\tilde{\boldsymbol{v}}^{\bar{t}}
+\boldsymbol{b}-\boldsymbol{w}^{\bar{t}})\big]
+\alpha\mathbb{E}_{\{\boldsymbol{r}^{\bar{t}}\}}
\big[\textstyle\sum\limits_{t=1}^{\bar{t}-1}
f(\boldsymbol{w}^{*})-f(\boldsymbol{w}^{t})-
\nonumber\\
&\textstyle
\quad\boldsymbol{\lambda}^T
(\boldsymbol{A}\tilde{\boldsymbol{v}}^{t}
+\boldsymbol{b}-\boldsymbol{w}^{t})\big]
%\label{convergence-25}
%\end{align}
%Regarding $[(\ref{convergence-25})-1]$ and $[(\ref{convergence-25})-2]$
%we have
%\begin{align}
\nonumber\\
%&\textstyle-([(\ref{convergence-25})-1]+[(\ref{convergence-25})-2])
&\textstyle\leq[(\ref{convergence-26})-1]-
\mathbb{E}_{\{\boldsymbol{r}^{\bar{t}}\}}
\Big[\big(f(\boldsymbol{w}^{\bar{t}})+
\alpha\sum_{t=1}^{\bar{t}-1}f(\boldsymbol{w}^{t})\big)-
\nonumber\\
&\textstyle
\big(1+\alpha\cdot(\bar{t}-1)\big)f(\boldsymbol{w}^{*})
+\boldsymbol{\lambda}^T(\boldsymbol{A}\tilde{\boldsymbol{v}}^{\bar{t}}
+\boldsymbol{b}-\boldsymbol{w}^{\bar{t}})
\nonumber\\
&\textstyle
+\alpha\sum_{t=1}^{\bar{t}-1}
\boldsymbol{\lambda}^T(\boldsymbol{A}\tilde{\boldsymbol{v}}^{t}
+\boldsymbol{b}-\boldsymbol{w}^{t})\Big]
\nonumber\\
&\textstyle \overset{(a)}{\leq} [(\ref{convergence-26})-1]
-\mathbb{E}_{\{\boldsymbol{r}^{\bar{t}}\}}
\big[\big(1+\alpha\cdot(\bar{t}-1)\big)\big( f(\vec{\boldsymbol{w}}^{\bar{t}})
-f(\boldsymbol{w}^{*})+
\nonumber\\
&\textstyle\qquad\qquad\qquad\qquad\qquad\boldsymbol{\lambda}^T
(\boldsymbol{A}\vec{\boldsymbol{v}}^{\bar{t}}
+\boldsymbol{b}-\vec{\boldsymbol{w}}^{\bar{t}})\big)\big],
\label{convergence-26}
\end{align}
where
\begin{align}
\textstyle\vec{\boldsymbol{w}}^{\bar{t}}\triangleq\frac{\boldsymbol{w}^{\bar{t}}+
\alpha\sum_{t=1}^{\bar{t}-1}\boldsymbol{w}^{t}}{1+\alpha\cdot(\bar{t}-1)}, \
\vec{\boldsymbol{v}}^{\bar{t}}\triangleq\frac{\tilde{\boldsymbol{v}}^{\bar{t}}+
\alpha\sum_{t=1}^{\bar{t}-1}\tilde{\boldsymbol{v}}^{t}}{1+\alpha\cdot(\bar{t}-1)},
\end{align}
%\begin{align}
%\textstyle \vec{\boldsymbol{w}}^{\bar{t}}\triangleq\frac{\boldsymbol{w}^{\bar{t}}+
%\alpha\sum_{t=1}^{\bar{t}-1}\boldsymbol{w}^{t}}{1+\alpha\cdot(\bar{t}-1)}, \
%\vec{\boldsymbol{v}}^{\bar{t}}\triangleq\frac{\tilde{\boldsymbol{v}}^{\bar{t}}+
%\alpha\sum_{t=1}^{\bar{t}-1}\tilde{\boldsymbol{v}}^{t}}{1+\alpha\cdot(\bar{t}-1)},
%\label{convergence-27}
%\end{align}
and $(a)$ comes from Jensen's inequality. From (\ref{convergence-26}) we see
\begin{align}
&\textstyle \mathbb{E}_{\{\boldsymbol{r}^{\bar{t}}\}}\big[
f(\vec{\boldsymbol{w}}^{\bar{t}})-f(\boldsymbol{w}^{*})+
\boldsymbol{\lambda}^T(\boldsymbol{A}\vec{\boldsymbol{v}}^{\bar{t}}
+\boldsymbol{b}-\vec{\boldsymbol{w}}^{\bar{t}})\big]
\nonumber\\
&\textstyle\leq \frac{\text{Const}+\frac{1}{2\beta}\|\boldsymbol{\lambda}^{0}-
\boldsymbol{\lambda}\|_2^2}{1+\alpha\cdot(\bar{t}-1)}.
\label{convergence-28}
\end{align}
Let $\boldsymbol{\lambda}$ be chosen as $\boldsymbol{\lambda}=
\vec{\boldsymbol{\lambda}}=2(\|\boldsymbol{\lambda}^*\|_2+\epsilon)\cdot
\frac{\boldsymbol{A}\vec{\boldsymbol{v}}^{\bar{t}}+\boldsymbol{b}-
\vec{\boldsymbol{w}}^{\bar{t}}}{\|\boldsymbol{A}\vec{\boldsymbol{v}}^{\bar{t}}
+\boldsymbol{b}-\vec{\boldsymbol{w}}^{\bar{t}}\|_2}$.
%\begin{align}
%\textstyle \boldsymbol{\lambda}=2(\|\boldsymbol{\lambda}^*\|_2+\epsilon)\cdot
%\frac{\boldsymbol{A}\vec{\boldsymbol{v}}^{\bar{t}}+\boldsymbol{b}-
%\vec{\boldsymbol{w}}^{\bar{t}}}{\|\boldsymbol{A}\vec{\boldsymbol{v}}^{\bar{t}}
%+\boldsymbol{b}-\vec{\boldsymbol{w}}^{\bar{t}}\|_2}
%\label{convergence-28-1}
%\end{align}
Then (\ref{convergence-28}) becomes
\begin{align}
&\textstyle \mathbb{E}_{\{\boldsymbol{r}^{\bar{t}}\}}\big[
\underbrace{f(\vec{\boldsymbol{w}}^{\bar{t}})-f(\boldsymbol{w}^{*})+
2(\|\boldsymbol{\lambda}^*\|_2+\epsilon)\cdot
\|\boldsymbol{A}\vec{\boldsymbol{v}}^{\bar{t}}
+\boldsymbol{b}-\vec{\boldsymbol{w}}^{\bar{t}}\|_2}_{[(\ref{convergence-29})-1]}\big]
\nonumber\\
&\textstyle
\leq \frac{\text{Const}+\frac{1}{2\beta}\|\boldsymbol{\lambda}^{0}-
\vec{\boldsymbol{\lambda}}\|_2^2}{1+\alpha\cdot(\bar{t}-1)}.
\label{convergence-29}
\end{align}
According to Theorem \ref{theorem-content-3} presented in
Appendix \ref{appendix-D}, we have
\begin{align}
&\textstyle 0\leq f(\vec{\boldsymbol{w}}^{\bar{t}})-f(\boldsymbol{w}^*)+
\langle\boldsymbol{\lambda}^*,\boldsymbol{A}\vec{\boldsymbol{v}}^{\bar{t}}
+\boldsymbol{b}-\vec{\boldsymbol{w}}^{\bar{t}}\rangle
\nonumber\\
&\textstyle
\leq \underbrace{f(\vec{\boldsymbol{w}}^{\bar{t}})-f(\boldsymbol{w}^*)+
\|\boldsymbol{\lambda}^*\|_2\|\boldsymbol{A}\vec{\boldsymbol{v}}^{\bar{t}}
+\boldsymbol{b}-\vec{\boldsymbol{w}}^{\bar{t}}\|_2}_{[(\ref{convergence-30})-1]}.
\label{convergence-30}
\end{align}
Using (\ref{convergence-30}), i.e., $[(\ref{convergence-30})-1]\geq 0$,
we have
\begin{align}
&[(\ref{convergence-29})-1] \geq \pm (f(\vec{\boldsymbol{w}}^{\bar{t}})-
f(\boldsymbol{w}^{*}) )
\nonumber\\
\Rightarrow& |f(\vec{\boldsymbol{w}}^{\bar{t}})-f(\boldsymbol{w}^{*})|
\leq [(\ref{convergence-29})-1];
\label{convergence-31}
\\
&[(\ref{convergence-29})-1] \geq \pm C\cdot \|\boldsymbol{A}\vec{\boldsymbol{v}}^{\bar{t}}
+\boldsymbol{b}-\vec{\boldsymbol{w}}^{\bar{t}}\|_2
\nonumber\\
\textstyle \Rightarrow & C\cdot \|\boldsymbol{A}\vec{\boldsymbol{v}}^{\bar{t}}
+\boldsymbol{b}-\vec{\boldsymbol{w}}^{\bar{t}}\|_2\leq [(\ref{convergence-29})-1],
\label{convergence-32}
\end{align}
where $C=\|\boldsymbol{\lambda}^*\|_2+\epsilon$. Substituting
(\ref{convergence-31}) and (\ref{convergence-32}) into
(\ref{convergence-29}) yields
\begin{align}
&\textstyle \mathbb{E}_{\{\boldsymbol{r}^{\bar{t}}\}}\big[
|f(\vec{\boldsymbol{w}}^{\bar{t}})-f(\boldsymbol{w}^{*})|\big]
\leq \frac{\text{Const}+\frac{1}{2\beta}\|\boldsymbol{\lambda}^{0}-
\vec{\boldsymbol{\lambda}}\|_2^2}{1+\alpha\cdot(\bar{t}-1)},
\nonumber\\
&\textstyle \mathbb{E}_{\{\boldsymbol{r}^{\bar{t}}\}}\big[
\|\boldsymbol{A}\vec{\boldsymbol{v}}^{\bar{t}}
+\boldsymbol{b}-\vec{\boldsymbol{w}}^{\bar{t}}\|_2\big]
\leq \frac{\text{Const}+\frac{1}{2\beta}\|\boldsymbol{\lambda}^{0}-
\vec{\boldsymbol{\lambda}}\|_2^2}{C(1+\alpha\cdot(\bar{t}-1))},
\label{convergence-33}
\end{align}
which is the desired results.

\section{Proving (\ref{convergence-17})}
\label{appendix-C}
Recall the famous parallelogram equation, i.e.,
\begin{align}
2\langle\boldsymbol{x}-\boldsymbol{y},\boldsymbol{x}-\boldsymbol{z}\rangle
=\|\boldsymbol{x}-\boldsymbol{y}\|_2^2+\|\boldsymbol{x}-\boldsymbol{z}\|_2^2
-\|\boldsymbol{y}-\boldsymbol{z}\|_2^2
\label{appendix-d-2}
\end{align}
Applying (\ref{appendix-d-2}) to the first two terms
$[(\ref{convergence-16})-1]$ yields
\begin{align}
&\textstyle (\bar{\boldsymbol{\lambda}}^{\bar{t}}-\boldsymbol{\lambda})^T
(\boldsymbol{\lambda}^{\bar{t}-1}-\bar{\boldsymbol{\lambda}}^{\bar{t}})
\nonumber\\
&\textstyle
=-\frac{1}{2}(\|\bar{\boldsymbol{\lambda}}^{\bar{t}}-\boldsymbol{\lambda}\|_2^2
+\|\boldsymbol{\lambda}^{\bar{t}-1}-\bar{\boldsymbol{\lambda}}^{\bar{t}}\|_2^2
-\|\boldsymbol{\lambda}^{\bar{t}-1}-\boldsymbol{\lambda}\|_2^2),
\label{appendix-d-1}
\end{align}
\begin{align}
&\textstyle (\boldsymbol{\lambda}^{t}-\boldsymbol{\lambda})^T
(\boldsymbol{\lambda}^{t-1}-\boldsymbol{\lambda}^{t})
\nonumber\\
&\textstyle
=-\frac{1}{2}(\|\boldsymbol{\lambda}^{t}-\boldsymbol{\lambda}\|_2^2
+\|\boldsymbol{\lambda}^{t-1}-\boldsymbol{\lambda}^{t}\|_2^2
-\|\boldsymbol{\lambda}^{t-1}-\boldsymbol{\lambda}\|_2^2),
\nonumber\\
&\quad 1\leq t\leq \bar{t}-1.
\label{appendix-d-1-1}
\end{align}
Similarly, for $\Xi(\boldsymbol{w}^{t})$, $1\leq t\leq \bar{t}$, we
also have
\begin{align}
&\textstyle (\boldsymbol{w}^{*}-\boldsymbol{w}^{t})^T
(\boldsymbol{w}^{t}-\boldsymbol{w}^{t-1})=
-\frac{1}{2}(\|\boldsymbol{w}^{t}-\boldsymbol{w}^*\|_2^2
+\|\boldsymbol{w}^{t-1}-\boldsymbol{w}^{t}\|_2^2
\nonumber\\
&\textstyle\qquad\qquad\qquad\qquad\qquad
-\|\boldsymbol{w}^{t-1}-\boldsymbol{w}^*\|_2^2).
\label{appendix-d-2-1}
\end{align}
Regarding the term $Y^{\bar{t}}$ we have
\begin{align}
&\textstyle Y^{\bar{t}}=\beta(\boldsymbol{w}^{\bar{t}}-\boldsymbol{w}^{\bar{t}-1})^T
\sum_{j=1}^K\boldsymbol{A}_j(\tilde{\boldsymbol{v}}_j^{*}-
\tilde{\boldsymbol{v}}_j^{\bar{t}})
\nonumber\\
&\textstyle\overset{(a)}{=}\beta(\boldsymbol{w}^{\bar{t}}-\boldsymbol{w}^{\bar{t}-1})^T
\big((\boldsymbol{w}^{*}-\boldsymbol{w}^{\bar{t}})+\frac{1}{\beta}
(\boldsymbol{\lambda}^{\bar{t}-1}-\bar{\boldsymbol{\lambda}}^{\bar{t}})\big)
\nonumber\\
&\textstyle\overset{(\ref{appendix-d-2})}{=}
(\boldsymbol{w}^{\bar{t}}-\boldsymbol{w}^{\bar{t}-1})^T
(\boldsymbol{\lambda}^{\bar{t}-1}-\bar{\boldsymbol{\lambda}}^{\bar{t}})
\nonumber\\
&\textstyle\qquad-\frac{\beta}{2}\big(\|\boldsymbol{w}^{\bar{t}}-\boldsymbol{w}^{\bar{t}-1}\|_2^2
+\|\boldsymbol{w}^{*}-\boldsymbol{w}^{\bar{t}}\|_2^2-
\|\boldsymbol{w}^{*}-\boldsymbol{w}^{\bar{t}-1}\|_2^2\big)
\nonumber\\
&\textstyle\overset{(b)}{\leq}
\frac{1}{2\beta}\|\boldsymbol{\lambda}^{\bar{t}-1}-
\bar{\boldsymbol{\lambda}}^{\bar{t}}\|_2^2
-\frac{\beta}{2}\big(\|\boldsymbol{w}^{*}-\boldsymbol{w}^{\bar{t}}\|_2^2-
\|\boldsymbol{w}^{*}-\boldsymbol{w}^{\bar{t}-1}\|_2^2\big),
\label{appendix-d-3}
\end{align}
where $(a)$ is deduced by
\begin{align}
&\textstyle\frac{1}{\beta}(\bar{\boldsymbol{\lambda}}^{\bar{t}}-
\boldsymbol{\lambda}^{\bar{t}-1})\overset{(\ref{sec2-2-1})}{=}
\boldsymbol{A}\tilde{\boldsymbol{v}}^{\bar{t}}
+\boldsymbol{b}-\boldsymbol{w}^{\bar{t}}
-(\boldsymbol{A}\tilde{\boldsymbol{v}}^{*}
+\boldsymbol{b}-\boldsymbol{w}^{*})
\nonumber\\
&\qquad\qquad\qquad\overset{\boldsymbol{A}\tilde{\boldsymbol{v}}^{*}
+\boldsymbol{b}-\boldsymbol{w}^{*}=\boldsymbol{0}}{=}
\textstyle\boldsymbol{A}(\tilde{\boldsymbol{v}}^{\bar{t}}
-\tilde{\boldsymbol{v}}^{*})-(\boldsymbol{w}^{\bar{t}}-\boldsymbol{w}^{*})
\nonumber\\
&\textstyle\Rightarrow\boldsymbol{A}(\tilde{\boldsymbol{v}}^{*}-
\tilde{\boldsymbol{v}}_j^{\bar{t}})=(\boldsymbol{w}^{*}-
\boldsymbol{w}^{\bar{t}})+
\frac{1}{\beta}(\boldsymbol{\lambda}^{\bar{t}-1}-
\bar{\boldsymbol{\lambda}}^{\bar{t}}),
\label{appendix-d-4}
\end{align}
and $(b)$ is obtained by applying the well-known inequality
$ab\leq \frac{\beta}{2}a^2+\frac{1}{2\beta}b^2$. As for $Y^{t}$,
$1\leq t\leq \bar{t}-1$, similar to the above deduction
we have
\begin{align}
&\textstyle Y^{t}\triangleq\beta(\boldsymbol{w}^{t}-\boldsymbol{w}^{t-1})^T
\sum_{j=1}^K\boldsymbol{A}_j(\tilde{\boldsymbol{v}}_j^{*}-
\tilde{\boldsymbol{v}}_j^{t})
\nonumber\\
&\textstyle
\overset{(a)}{=}\beta(\boldsymbol{w}^{t}-\boldsymbol{w}^{t-1})^T
\big((\boldsymbol{w}^{*}-\boldsymbol{w}^{t})+\frac{1}{\alpha\beta}
(\boldsymbol{\lambda}^{t-1}-\boldsymbol{\lambda}^{t})\big)
\nonumber\\
&\textstyle =\frac{1}{\alpha}(\boldsymbol{w}^{t}-\boldsymbol{w}^{t-1})^T
(\boldsymbol{\lambda}^{t-1}-\boldsymbol{\lambda}^{t})-
\nonumber\\
&\textstyle
\qquad\frac{\beta}{2}\big(\|\boldsymbol{w}^{t}-\boldsymbol{w}^{t-1}\|_2^2
+\|\boldsymbol{w}^{*}-\boldsymbol{w}^{t}\|_2^2-
\|\boldsymbol{w}^{*}-\boldsymbol{w}^{t-1}\|_2^2\big)
\nonumber\\
&\textstyle\overset{(b)}{\leq}
\frac{1}{2\beta}\|\boldsymbol{\lambda}^{t-1}-\boldsymbol{\lambda}^{t}\|_2^2+
(\frac{\beta}{2\alpha^2}-
\frac{\beta}{2})\|\boldsymbol{w}^{t}-\boldsymbol{w}^{t-1}\|_2^2-
\nonumber\\
&\textstyle
\qquad
\frac{\beta}{2}\big(\|\boldsymbol{w}^{*}-\boldsymbol{w}^{t}\|_2^2-
\|\boldsymbol{w}^{*}-\boldsymbol{w}^{t-1}\|_2^2\big),
\label{appendix-d-5}
\end{align}
where $(a)$ is because $\sum_{j=1}^K \boldsymbol{A}_j(\tilde{\boldsymbol{v}}_j^{*}-
\tilde{\boldsymbol{v}}_j^{t})=(\boldsymbol{w}^{*}-\boldsymbol{w}^{t})
+\frac{1}{\alpha\beta}(\boldsymbol{\lambda}^{t-1}-\boldsymbol{\lambda}^{t})$,
%\begin{align}
%&\textstyle\frac{1}{\alpha\beta}(\boldsymbol{\lambda}^{t}-
%\boldsymbol{\lambda}^{t-1})\overset{(\ref{sec2-2})}{=}
%\textstyle\sum_{j=1}^K \boldsymbol{A}_j(\tilde{\boldsymbol{v}}_j^{t}
%-\tilde{\boldsymbol{v}}_j^{*})-(\boldsymbol{w}^{t}-\boldsymbol{w}^{*})
%\nonumber\\
%&\textstyle\Rightarrow\sum_{j=1}^K \boldsymbol{A}_j(\tilde{\boldsymbol{v}}_j^{*}-
%\tilde{\boldsymbol{v}}_j^{t})=(\boldsymbol{w}^{*}-\boldsymbol{w}^{t})
%+\frac{1}{\alpha\beta}(\boldsymbol{\lambda}^{t-1}-\boldsymbol{\lambda}^{t}),
%\label{appendix-d-6}
%\end{align}
and $(b)$ has invoked again $\frac{1}{\alpha}\cdot
ab\leq \frac{\beta}{2\alpha^2}a^2+\frac{1}{2\beta}b^2$. Substituting
(\ref{appendix-d-1}), (\ref{appendix-d-1-1}), (\ref{appendix-d-2-1}),
(\ref{appendix-d-3}) and (\ref{appendix-d-5}) into
$[(\ref{convergence-16})-1]$ yields the desired result.

\section{Theorem \ref{theorem-content-3}}
\label{appendix-D}
\newtheorem{theorem2}{Theorem}
\begin{theorem}[\cite{HeYuan12}]
\label{theorem-content-3}
Consider the following linearly constrained problem:
\begin{align}
\mathop {\min}\limits_{\boldsymbol{x},\boldsymbol{y}} & \
\textstyle f(\boldsymbol{y})
\nonumber\\
\text{s.t.} & \ \textstyle \boldsymbol{A}\boldsymbol{x}
+\boldsymbol{y}=\boldsymbol{b}.
\label{theorem-HeYuan12}
\end{align}
Assume that $f(\boldsymbol{y})$ is convex and continuously differentiable.
Let $\{\boldsymbol{x}^*,\boldsymbol{y}^*,\boldsymbol{\lambda}^*\}$
denote a set of Lagrangian primal-dual solution to (\ref{problem-linear-cons2}).
Then it holds
\begin{align}
\textstyle f(\boldsymbol{y})-f(\boldsymbol{y}^*)+\langle\boldsymbol{\lambda}^*,
\boldsymbol{A}\boldsymbol{x}+\boldsymbol{y}-\boldsymbol{b}\rangle\geq 0,
\ \forall \boldsymbol{x}, \  \boldsymbol{y}.
\label{appendix-e-1}
\end{align}
\end{theorem}

\bibliography{newbib}
\bibliographystyle{IEEEtran}
\end{document}